\PassOptionsToPackage{dvipsnames}{xcolor}
\documentclass[runningheads]{llncs}
\usepackage{booktabs}   

\usepackage{stmaryrd}
\usepackage{amssymb}
\usepackage{color}
\usepackage{listings}
\usepackage{mathpartir}
\usepackage{pifont}

\usepackage{tcolorbox}
\usepackage{tabularx}
\usepackage{xcolor}
\usepackage{wrapfig}
\usepackage{bm}

\colorlet{lprolog}{blue!50!blue}  
\colorlet{abellatop}{blue!70!green}
\colorlet{abellatac}{orange!30!black}
\colorlet{abellabad}{red!80!yellow}

\lstdefinelanguage{lprolog}{%
  alsoletter={-},
  classoffset=0,%
  morekeywords={sig,module,type,kind,pi,sigma,end,true,false,remove,from,import,if,then,else,=,{!,!},let,U,\>},%
  keywordstyle=\color{lprolog},%
  classoffset=0,%
  otherkeywords={:-,\&,\{,\},!},%
  sensitive=true,%
  morestring=[bd]",%
  morecomment=[l]\%,%
  morecomment=[n]{(*}{*)},%
}

\lstdefinelanguage{abella}[]{lprolog}{%
  alsoletter={-},
  classoffset=1,%
  morekeywords={Close,CoDefine,Define,Kind,Query,Quit,Specification,
    Set,Split,Theorem,Undo,by,as,prop,forall,exists,nabla},%
  keywordstyle=\color{abellatop},%
  classoffset=2,%
  morekeywords={abbrev,apply,backchain,case,coinduction,cut,
    induction,inst,intros,monotone,on,permute,rename,left,right,witness,
    search,split,to,unabbrev,unfold,assert,with},%
  keywordstyle=\color{abellatac},%
  classoffset=3,%
  morekeywords={undo,abort,skip,clear},%
  keywordstyle=\color{abellabad}\underbar,%
  classoffset=0,%
}

\newcommand{\consLNCFilter}[2]{\ifLang{(\isNothing #1)}{#2}{\key{cons} \app (\getLNC{#1}) \app #2}}
\newcommand{\consLNCFilterName}{\mathit{cons}^{*}}

\newcommand{\LNC}[1]{{\hat{#1}}}

\newcommand{\isIn}[2]{#1\app \key{in} \app #2}
\newcommand{\isEmptyLNC}[1]{\key{isEmpty} \app #1}
\newcommand{\isNothing}[1]{\key{isNothing} \app #1}

\newcommand{\LangDef}{\mathcal{L}}

\newcommand{\calcName}{\mathcal{L}\textendash\textsf{Tr}}

\newcommand{\matchLNC}[2]{\mathit{match}(#1,#2)}

\newcommand{\errorLNC}{\key{error}}

\newcommand{\typeOfPattern}{\typeOf}

\newcommand{\atomic}{_\key{@}}		
\newcommand{\subRule}{_\textsf{r}}		
\newcommand{\subFormula}{_\textsf{f}}		
\newcommand{\subListFormula}{_\textsf{lf}}		
\newcommand{\subTerm}{_\textsf{t}}		
\newcommand{\subFormulaAndTerm}{_{(\textsf{f } and \textsf{ t})}}
\newcommand{\subListTerm}{_\textsf{lt}}

\newcommand{\ifSEM}{\mathit{if}}
\newcommand{\thenSEM}{\mathit{then}}
\newcommand{\elseSEM}{\mathit{else}}
\newcommand{\letSEM}{\mathit{let}}
\newcommand{\inSEM}{\mathit{in}}
\newcommand\eqSEM{\app %
 \rule[.4ex]{5pt}{0.5pt}\llap{\rule[.8ex]{5pt}{0.5pt}}\app }

\newcommand{\uniquefySEMJustName}{\mathit{uniquefy}}

\newcommand{\uniquefySEM}[5]{\mathit{uniquefy}#1(#2,#3,#4, #5)}
\newcommand{\uniquefySEMfound}[4]{\mathit{uniquefy}^{\bullet}(#1, #2,#3,#4)}
\newcommand{\uniquefySEMreplace}[3]{\mathit{uniquefy}^{\dagger}#1(#2,#3)}

\newcommand{\castLNC}{\key{cast}}
\newcommand{\consLNC}[2]{\key{cons} \app #1 \app #2}

\newcommand{\emptyListLNC}{\key{nil}}
\newcommand{\emptyMapLNC}{\key{emptyMap}}
\newcommand{\setRulesLNC}[1]{\key{setRules} \app #1}
\newcommand{\getRulesLNC}{\key{getRules}}

\newcommand{\tickSEMna}{\mathit{tick}}
\newcommand{\tickSEM}[1]{\mathit{tick}(#1)}
\newcommand{\varsSEM}[1]{\mathit{vars}(#1)}
\newcommand{\rangeSEM}[1]{\mathit{range}(#1)}
\newcommand{\zipSEM}[2]{\mathit{zip}(#1,#2)}
\newcommand{\lastSEM}[1]{\mathit{last}(#1)}
\newcommand{\checkzipSEM}[2]{\mathit{checkZip}(#1,#2)}
\newcommand{\failureSEM}{\mathit{fail}}

\newcommand{\LanguageType}{\key{Language}}
\newcommand{\RuleType}{\key{Rule}}

\newcommand{\FormulaType}{\key{Formula}}
\newcommand{\TermType}{\key{Term}}
\newcommand{\VarType}{\key{MetaVar}}
\newcommand{\MapType}[2]{\key{Map} \app #1 \app #2}
\newcommand{\StringType}{\key{String}}

\newcommand{\OpnameType}{\key{OpName}}
\newcommand{\PrednameType}{\key{PredName}}
\newcommand{\NameType}{\key{Name}}
\newcommand{\MaybeType}[1]{\key{Option} \app #1}

\newcommand{\langVar}[1]{{#1}}

\newcommand{\getOverlap}[2]{\key{varOverlap}\app(#1,#2)}
\newcommand{\newVar}{\key{newVar}}
\newcommand{\varsLNC}[1]{\key{vars}(#1)}

\newcommand{\createMapLNC}[2]{\key{map}(#1,#2)}
\newcommand{\tickedLNC}[1]{#1\texttt{'}}
\newcommand{\unboundLNC}[1]{\key{unbound}(#1)}

\newcommand{\tickedMMLNC}[2]{\tickedLNC{#1}}

\newcommand{\justLNC}[1]{\key{just}\app#1}
\newcommand{\getLNC}[1]{\key{get}\app#1}
\newcommand{\filterNothingLNC}[1]{\mathit{filterNothing}(#1)}

\newcommand{\GammaForCnames}[1]{\Gamma_{\textsf{cn}}^{(#1)}}
\newcommand{\GammaForRule}[1]{\Gamma_{\textsf{rule}}}

\newcommand{\subsForRule}[1]{\theta_{\textsf{rule}}^{(#1)}}
\newcommand{\subsForCnames}[1]{\theta_{\textsf{cn}}^{(#1)}}
\newcommand{\ruleComposition}[2]{#1 ;_{\textsf{r}} #2}

\newcommand{\eliminatorOf}[2]{\key{eliminatorOf}(#1,#2)}
\newcommand{\andLNC}[2]{#1 \app \key{and}\app #2}
\newcommand{\orLNC}[2]{#1 \app \key{or}\app #2}
\newcommand{\notLNC}[1]{\key{not}\app #1}

\newcommand{\foldLNC}[2]{\key{fold}\app #1 \app #2}

\newcommand{\mapKeysLNC}[1]{\key{mapKeys}\app #1}

\newcommand{\conclLNC}{\mathit{conclusion}}
\newcommand{\premisesLNC}{\mathit{premises}}

\newcommand{\selfLNC}{\mathit{self}}

\newcommand{\uniquefy}[6]{\key{uniquefy}(#1,#2,#3) \Rightarrow (#4,#5):#6 }
\newcommand{\isContra}{\key{contra}}
\newcommand{\isCovar}{\key{covariant}}
\newcommand{\isInvar}{\key{invariant}}
\newcommand{\inductive}{\key{inductive}}
\newcommand{\isInput}{\key{input}}
\newcommand{\overlap}{\key{varOverlap}}
\newcommand{\skipLNC}{\key{skip}}
\newcommand{\nothingLNC}{\key{nothing}}

\newcommand{\select}[3]{{#1}[#2]:\app #3}

\newcommand{\runningExample}{\lambda^{\Bool}}
 
\newcommand{\langCom}{\textsc{lang-n-change}}
\newcommand{\forLang}[2]{\key{for \app  each} \app #1 \app \key{in} \app #2:\app}
\newcommand{\forOrderedLang}[3]{\key{for \app  ordered} \app #1 \app \app #2 \app \key{in} \app #3:\app}
\newcommand{\ifLang}[3]{\key{if} \app #1 \app \key{then} \app #2 \app \key{else}\app #3}

\newcommand{\ruletag}[1]{\textsc{(#1)}}

\newcommand{\grammarDerivation}[1]{{\Rightarrow_{G'}^{*}}}
\newcommand{\grammarDerivationG}[1]{{\Rightarrow_{G}^{*}}}

\definecolor{magenta(dye)}{rgb}{0.79, 0.08, 0.48}
\definecolor{bondiblue}{rgb}{0.0, 0.58, 0.71}
\definecolor{navyblue}{rgb}{0.0, 0.0, 0.5}
\definecolor{lightskyblue}{rgb}{0.53, 0.81, 0.98}

\newcommand*\colourcheck[1]{%
  \expandafter\newcommand\csname #1check\endcsname{\textcolor{#1}{\text{\ding{52}}}}%
}
\colourcheck{blue}
\colourcheck{green}
\colourcheck{red}

\newcolumntype{Y}{>{\raggedleft\arraybackslash}X}

\tcbset{tab1/.style={fonttitle=\bfseries\large,fontupper=\normalsize\sffamily,
colback=yellow!10!white,colframe=red!75!black,colbacktitle=red!40!white,
coltitle=black,center title,freelance,frame code={
\foreach \n in {north east,north west,south east,south west}
{\path [fill=red!75!black] (interior.\n) circle (3mm); };},}}

\tcbset{tab2/.style={enhanced,fonttitle=\bfseries,fontupper=\small\sffamily,
colback=yellow!10!white,colframe=red!50!black,colbacktitle=red!40!white,
coltitle=black,center title}}

\usepackage{mathtools}
\usepackage{relsize}

\newcommand{\ninference}[3]{\inferrule[(#1)]{#2}{#3}}

\newcommand{\ba}{\begin{array}}
\newcommand{\ea}{\end{array}}

\newenvironment{syntax}{\[\ba{l@{\;\;}lcl}}{\ea\]}

\definecolor{ShadowColor}{RGB}{30,150,190}

\makeatletter
\newcommand\Cshadowbox{\VerbBox\@Cshadowbox}
\def\@Cshadowbox#1{%
  \setbox\@fancybox\hbox{\fbox{#1}}%
  \leavevmode\vbox{%
    \offinterlineskip
    \dimen@=\shadowsize
    \advance\dimen@ .5\fboxrule
    \hbox{\copy\@fancybox\kern.5\fboxrule\lower\shadowsize\hbox{%
      \color{ShadowColor}\vrule \@height\ht\@fancybox \@depth\dp\@fancybox \@width\dimen@}}%
    \vskip\dimexpr-\dimen@+0.5\fboxrule\relax
    \moveright\shadowsize\vbox{%
      \color{ShadowColor}\hrule \@width\wd\@fancybox \@height\dimen@}}}
\makeatother

\newcommand{\matches}{\mathrel{\triangleright}}

\newcommand{\typeOf}{\vdash}
\newcommand{\step}{\longrightarrow}

\definecolor{lightblue}{rgb}{0.25,0.25,1}

\definecolor{lightgray}{gray}{0.9}
\definecolor{darkergrey}{rgb}{0.75, 0.75, 0.75}
\newcommand{\HI}[1]{\colorbox{darkergrey}{#1}}

\newcommand{\key}[1]{\ensuremath{\mathtt{#1}}}

\newcommand{\Int}{\key{Int}}
\newcommand{\Bool}{\key{Bool}}
\newcommand{\Float}{\key{Float}}

\newcommand{\app}{\;}

\newcommand{\If}{\key{if}\:}

\newcommand{\tagsc}[1]{\tag{\textsc{#1}}}

\definecolor{lightgray}{gray}{0.9}

\newcommand\dyn{\key{dyn}}

\newcommand{\true}{{\key{true}}}

\newcommand{\List}{{\key{List}}}

\newcommand{\nil}{{\key{nil}}}
\newcommand{\cons}{{\key{cons}}}
\newcommand{\head}{{\key{head}}}
\newcommand{\tail}{{\key{tail}}}

\newcommand{\Ldl}{\mathcal{L}}

\newcommand{\emptyLdl}[1]{\Ldl^{\epsilon}}

\usepackage{semantic}

\title{A Calculus for Language Transformations}
\author{Benjamin Mourad\and
Matteo Cimini}
\authorrunning{Mourad and Cimini}
\institute{University of Massachusetts Lowell, Lowell MA 01854, USA 
}
\begin{document}

\maketitle              
\begin{abstract}
In this paper we propose a calculus for expressing algorithms for programming languages transformations. 
We present the type system and operational semantics of the calculus, and we prove that it is type sound. 
We have implemented our calculus, and we demonstrate its applicability with common examples in programming languages.
As our calculus manipulates inference systems, our work can, in principle, be applied to logical systems. 
\end{abstract}

\section{Introduction}\label{intro}

Operational semantics is a standard de facto to defining the semantics of programming languages \cite{PLOTKIN}. 
However, producing a programming language definition is still a hard task. It is not surprising that theoretical and software tools for supporting the modeling of languages based on operational semantics have received attention in research \cite{LangWorkbenches,Rosu2010,Redex}. 
In this paper, we address an important aspect of language reuse which has not received attention so far: Producing language definitions from existing ones by the application of transformation algorithms. Such algorithms may automatically add features to the language, or switch to different semantics styles. In this paper, we aim at providing theoretical foundations and a software tool for this aspect. 

Consider the typing rule of function application below on the left and its version with algorithmic subtyping on the right. 
{\footnotesize 
\begin{gather*}
{
\ninference{t-app}
	{
	\Gamma \typeOf \app e_1 : T_1\to T_2 \\
	\Gamma \typeOf \app e_2 : T_1 	
	} 
	{ \Gamma \typeOf \app e_1\app e_2 : T_2}
}
 ~~ \stackrel{f(\textsc{t-app})}{\Longrightarrow}  ~~
{
\ninference{t-app'}
	{
	\Gamma \typeOf \app e_1 : T_{11}\to T_2 \\
	\Gamma \typeOf \app e_2 : T_{12} \\\\  T_{12} <: T_{11}
	} 
	{ \Gamma \typeOf \app e_1\app e_2 : T_2}
}
\end{gather*}
}
Intuitively, we can describe \textsc{(t-app')} as a function of \textsc{(t-app)}. Such a function includes, at least, giving new variable names when a variable is mentioned more than once, and must relate the new variables with subtyping according to the variance of types (covariant vs contravariant). 
Our question is: \emph{Can we express, easily, language transformations in a safe calculus?} 

Language transformations are beneficial for a number of reasons. On the theoretical side, they isolate and make explicit the insights that underly some programming languages features or semantics style. On the practical side, language transformations do not apply just to one language but to several languages. They can alleviate the burden to language designers, who can use them to automatically generate new language definitions using well-established algorithms rather than manually defining them, an error prone endeavor. 

%
%
In this paper, we make the following contributions. 
\begin{itemize}
\item We present $\calcName$ (pronounced ``Elter''), a formal calculus for language transformations (Section \ref{main}). We define the syntax (Section \ref{syntax}), operational semantics (Section \ref{operational}), and type system (Section \ref{typesystem}) of $\calcName$. 
\item We prove that $\calcName$ is type sound (Section \ref{typesystem}). 
\item We show the applicability of $\calcName$ to the specification of two transformations: adding subtyping and switching from small-step to big-step semantics (Section \ref{examples}). Our examples show that $\calcName$ is expressive and offers a rather declarative style to programmers. 
\item We have implemented $\calcName$ \cite{ltr}, and we report that we have applied our transformations to several language definitions.
 \end{itemize}

Related work are discussed in Section \ref{related}, and Section \ref{conclusion} concludes the paper. 
\section{A Calculus for Language Transformations}\label{main}

We focus on language definitions in the style of operational semantics. 
To briefly summarize, languages are specified with a BNF grammar and a set of inference rules. 
BNF grammars have \emph{grammar productions} such as $\text{\sf Types} \app T  ::=  B \mid \app T\to  T$. We call $\text{\sf Types}$ a \emph{category name}, $T$ is a \emph{grammar meta-variable}, and $B$ and $T\to  T$, as well as, for example, $(\lambda x.e\app v)$, are \emph{terms}. 
$(\lambda x.e\app v) \step e[v/x]$ and $\Gamma \typeOf (e_1\app e_2) : T_2$ are \emph{formulae}. 
An \emph{inference rule} $\inference{f_1, \ldots, f_n}{f}$ has a set of formulae above the horizontal line, which are called \emph{premises}, and a formula below the horizontal line, which is called the \emph{conclusion}. 



\subsection{Syntax of $\calcName$}\label{syntax}

Below we show the $\calcName$ syntax for language definitions, which reflects 
the operational semantics style of defining languages. Sets are accommodated with lists.  
%




{\footnotesize
$
  cname \in \textsc{CatName}, \langVar{X} \in \textsc{Meta-Var}, 
  opname \in \textsc{OpName}, predname \in \textsc{PredName}
$
\begin{syntax}
   \text{\sf Language} & \LangDef & ::= & (G,R) \\
   \text{\sf Grammar} & G & ::= & \{ s_1, \ldots, s_n \} \\
   \text{\sf Grammar Pr.} & s & ::= & cname \app \langVar{X} ::= lt\\
   \text{\sf Rule} & r & ::= & \inference{lf}{f}\\
   \text{\sf Formula} & f & ::= & predname \app lt  \\
   \text{\sf Term} & t & ::= &  \langVar{X} \mid opname\app lt  \mid (\langVar{X})t  \mid t[t/\langVar{X}] \\
   \text{\sf List of Rules} & R & ::= & \nil \mid \consLNC{r}{R}\\
   \text{\sf List of Formula} & \mathit{lf} & ::= & \nil \mid \consLNC{f}{\mathit{lf}}\\
   \text{\sf List of Terms} & lt & ::= & \nil \mid \consLNC{t}{lt}
\end{syntax}
}

We assume a set of category names \textsc{CatName}, a set of meta-variables \textsc{Meta-Var}, a set of constructor operator names \textsc{OpName}, and a set of predicate names \textsc{PredName}. We assume that these sets are pairwise disjoint. 
\textsc{OpName} contains elements such as $\to$ and $\lambda$ (elements do not have to necessarily be (string) names). 
\textsc{PredName} contains elements such as $\typeOf$ and $\step$.
To facilitate the modeling of our calculus, we assume that terms and formulae are defined in abstract syntax tree fashion. Here this means that they always have a top level constructor applied to a list of terms. 
%
%
$\calcName$ also provides syntax to specify unary binding $(\langVar{z})t$ and capture-avoiding substitution $t[t/\langVar{z}]$. Therefore, $\calcName$ is tailored for static scoping rather than dynamic scoping. 
Lists can be built as usual with the $\nil$ and \key{cons} operator. We sometimes use the shorthand $[o_1, \ldots o_n]$ for the corresponding series of $\key{\cons}$ applications ended with $\nil$. 

To make an example, the typing rule for function application and the $\beta$-reduction rules are written as follows. ($app$ is the top-level operator name for function application).
{\footnotesize
\begin{gather*}
\inference{[\app \typeOf \app [\langVar{\Gamma}, \langVar{e_1}, (\to [\langVar{T_1}, \langVar{T_2}])], \app \typeOf \app [\langVar{\Gamma}, \langVar{e_2}, \langVar{T_1}]\app ]} 
	{ \typeOf \app [\langVar{\Gamma}, (app \app [\langVar{e_1}, \langVar{e_2}]), \langVar{T_2}]}
\qquad
\inference{[]}{\step \app [(app \app [(\lambda \app [(\langVar{x})\langVar{e}]), \langVar{v}]), \langVar{e}[\langVar{v}/x]] }
\end{gather*}
}
%
%

Below we show the rest of the syntax of $\calcName$. \\
{\footnotesize
$
  x \in \textsc{Var}
  , str \in \textsc{String}
  , \{\selfLNC, \premisesLNC, \conclLNC\} \subseteq \textsc{Var}
$
\begin{syntax}
   \text{\sf Expression} & e & ::= & x \mid cname \mid str \mid \LNC{t} \mid \LNC{f} \mid \LNC{r}\\
   &&& \mid  \emptyListLNC \mid  \consLNC{e}{e}\mid \head \app e \mid \tail \app e \mid e @ e \\
   &&& \mid \createMapLNC{e}{e} \mid  e(e)\mid  \mapKeysLNC{e}\\
   &&& \mid \justLNC{e} \mid  \nothingLNC \mid \getLNC{e}\\
   &&& \mid cname \app \langVar{X} ::= e  \app  \mid cname \app \langVar{X} ::= \ldots \app e \\ 
   &&&\mid \getRulesLNC\mid \setRulesLNC{e}\\
   &&& \mid \select{e}{p}{e}   \mid \select{e(\key{keep})}{p}{e} 
   \mid \uniquefy{e}{e}{str}{x}{x}{e} \\
   &&& \mid \ifLang{b}{e}{e} \mid e \app ; \app e \mid \ruleComposition{e}{e} \mid  \skipLNC  \\
   &&&  \mid  \newVar \mid \tickedMMLNC{e}{e}\mid \foldLNC{predname}{e} \\ 
   &&&  \mid \errorLNC\\
   \text{\sf Boolean Expr.} & b & ::= &
   e == e \mid \isEmptyLNC{e} \mid \isIn{e}{e}\mid \isNothing{e} 
   \mid \andLNC{b}{b} \mid \orLNC{b}{b} \mid \notLNC{b} \\
   \text{\sf $\calcName$ Rule} & \LNC{r} & ::= & \inference{e}{e}\\
   \text{\sf $\calcName$ Formula} & \LNC{f} & ::= & predname \app e \mid x \app e \\
   \text{\sf $\calcName$ Term} & \LNC{t} & ::= &  \langVar{X} \mid opname\app e \mid x \app e \mid (\langVar{X})e  \mid e[e/\langVar{X}] \\
   \text{\sf Pattern} & p & ::= & x:T\mid predname \app p \mid opname \app p \mid x \app p \mid \emptyListLNC \mid  \consLNC{p}{p}\\
   \text{\sf Value} & v & ::= & {t} \mid {f} \mid {r} \mid cname\mid str \\
   &&& \mid  \emptyListLNC \mid  \consLNC{v}{v} 
  \mid  \createMapLNC{v}{v}\mid  \justLNC{v} \mid  \nothingLNC \mid \skipLNC
   \end{syntax}
%
}
Programmers write expressions to specify transformations. 
At run-time, an expression will be executed with a language definition. 
Evaluating an expression may modify the current language definition. 

\emph{Design Principles:} 
We strive to offer well-crafted operations that map well with the language manipulations that are frequent in adding features to languages or switching semantics styles. 
There are three features that we can point out which exemplify our approach the most: 1) The ability to program parts of rules, premises and grammars, 2) selectors $\select{e}{p}{e}$, and 3) the \key{uniquefy} operation. 
Below, we shall the describe the syntax for transformations, and place some emphasis in motivating these three operations. 
\\ \indent \emph{Basic Data Types:} 
$\calcName$ has strings and has lists with typical operators for extracting their head and tail, as well as for concatenating them ($@$). 
$\calcName$ also has maps (key-value). In $\createMapLNC{e_1}{e_2}$, $e_1$ and $e_2$ are lists. The first element of $e_1$ is the key for the first element of $e_2$, and so on for the rest of elements. Such a representation fits better our language transformations examples, as we shall see in Section \ref{examples}.  Operation $e_1(e_2)$  queries a map, where $e_1$ is a map and $e_2$ is a key, and $\mapKeysLNC{e}$ returns the list of keys of the map $e$. 
Maps are convenient in $\calcName$ to specify information that is not expressible in the language definition. For example, we can use maps to store information about whether some type argument is covariant or contravariant, or to store information about the input-output mode of the arguments of relations. Section \ref{examples} shows that we use maps in this way extensively. 
$\calcName$ also has options (\key{just}, \key{nothing}, and \key{get}). We include options because they are frequently used in combination with the selector operator described below. 
Programmers can refer to grammar categories (\emph{cname}) in positions where a list is expected. When \emph{cname} is used the corresponding list of grammar items is retrieved.
%
\\ \indent \emph{Grammar Instructions}:
$cname \app \langVar{X} ::= e$ is essentially a grammar production. With this instruction, the current grammar is augmented with this production.  
$cname \app \langVar{X} ::= \ldots \app e$ (notice the dots) adds the terms in $e$ to an existing production. 
$\getRulesLNC$ and $\setRulesLNC{e}$ retrieve and set the current list of rules, respectively. 
%
\\ \indent \emph{Selectors}:
$\select{e_1}{p}{e_2}$ is the selector operator. 
This operation selects one by one the elements of the list $e_1$ that satisfy the pattern $p$ and executes the body $e_2$ for each of them. 
This operation returns a list that collects the result of each iteration. 
Selectors are useful for selecting elements of a language with great precision, and applying manipulations to them. To make an example, suppose that the variable \emph{prems} contains the premises of a rule and that we wanted to invert the direction of all subtyping premises in it. 
The operation $\select{prems}{T_1 <: T_2}{\key{just} \app T_2 <: T_1}$ does just that. 
Notice that the body of a selector is an option. This is because it is common for some iteration to return no values ($\nothingLNC$). The examples in Section \ref{examples} show this aspect. 
Since options are commonly used in the context of selector iterations, we have designed our selector operation to automatically handle them. That is, $\nothingLNC$s are automatically removed, and the selector above returns the list of new subtyping premises rather than a list of options. 
The selector $\select{e(\key{keep})}{p}{e}$ works like an ordinary selector except that it also returns the elements that failed the pattern-matching.  
%
\\ \indent\emph{Uniquefy}:
When transforming languages it is often necessary to assign distinct variables. The example of algorithmic subtyping in the introduction is archetypal. 
$\calcName$ accommodates this operation as primitive with \key{uniquefy}. \\
$\uniquefy{e_1}{e_2}{str}{x}{y}{e_3}$ takes in input a list of formulae $e_1$, a map $e_2$, and a string $str$ (we shall discuss $x$, $y$, and $e_3$ shortly). This operation modifies the formulae $e_2$ to use different variable names when a variable is mentioned more than once. However, not every variable is subject to the replacement. Only the variables that appear in some specific positions are targeted. 
The map $e_2$ and the string $str$ contain the information to identify these positions. 
$e_2$ maps operator names and predicate names to a list that contains a label (as a string) for each of their arguments. 
For example, the map $m = \{\typeOf \app \mapsto [``{in}", ``{in}",  ``{out}"]\}$ says that $\Gamma$ and $e$ are inputs in a formula $\Gamma \typeOf e:T$, and that $T$ is the output. 
Similarly, the map $\{\to\app  \mapsto [``{contravariant}", ``{covariant}""]\}$ says that $T_1$ is contravariant and $T_2$ is covariant in $T_1 \to T_2$. 
The string $str$ specifies a label. $\calcName$ inspects the formulae in $e_1$ and their terms. Arguments that correspond to the label according to the map 
then receive a new variable. 
%
To make an example, if $\mathit{lf}$ is the list of premises of \textsc{(t-app)} and $m$ is defined as above (input-output modes), the operation $\uniquefy{\mathit{lf}}{m}{``{out}"}{x}{y}{e_3}$ creates the premises of \textsc{(t-app')} shown in the introduction. 
Furthermore, the computation continues with the expression $e_3$ in which $x$ is bound to these premises and $y$ is bound to a map that summarizes the changes made by \key{uniquefy}. 
This latter map associates every variable $X$ to the list of new variables that \key{uniquefy} has used to replace $X$. For example, since \key{uniquefy} created the premises of \textsc{(t-app')} by replacing $T_1$ in two different positions with $T_{11}$ and $T_{12}$, the map $\{T_1\mapsto [T_{11}, T_{12}]\}$ is passed to $e_3$ as $y$. 
Section \ref{examples} will show two examples that make use of \key{uniquefy}.
%
%
\\\indent \emph{Control Flow}:
$\calcName$ includes the if-then-else statement with typical guards.  
$\calcName$ also has the sequence operation $;$ (and $\skipLNC$) to execute language transformations one after another. $e_1 ;_{\text{r}} e_2$, instead, executes sequences of transformations on rules. After $e_1$ evaluates to a rule, $e_2$ makes use of that rule as the subject of its transformations. 
%
\\\indent\emph{Programming Rules, Premises, and Terms}:
In $\calcName$ a programmer can write $\calcName$ terms ($\hat{t}$), $\calcName$ formulae ($\hat{f}$), and $\calcName$ rules ($\hat{r}$) in expressions. 
These differ from the terms, formulae and rules of language definitions in that they can contain arbitrary expressions, such as if-then-else statements, at any position. 
This is a useful feature as it provides a declarative way to create rules, premises, or terms. To make an example with rule creation, we can write 
\[\inference{\select{prems}{T_1 <: T_2}{\key{just} \app T_2 <: T_1}}{f}\]
where \emph{prems} is the list of premises from above, and $f$ is a formula. As we can see, using expressions above the horizontal line is a convinient way to compute the premises of a rule. 
\\\indent\emph{Other Operations}:
The operation $\foldLNC{predname}{e}$ creates a list of formulae that interleaves $predname$ to any two subsequent elements of the list $e$. 
To make an example, the operation $\foldLNC{=}{[T_1, T_2, T_3, T_4]}$ generates the list of formulae $[T_1 = T_2, T_2 = T_3, T_3 = T_4]$. 
$\varsLNC{e}$ returns the list of the meta-variables in $e$. $\newVar$ returns a meta-variable that has not been previously used. The tick operator $\tickedMMLNC{e}{}$ gives a prime $'$ to the meta-variables of $e_1$ ($\langVar{X}$ becomes $\langVar{X'}$). 
$\key{vars}$ and the tick operator also work on lists of terms. 
\\\indent\emph{Variables and Substitution:}
Some variables have a special treatment in $\calcName$. 
We can refer to the value that a selector iterates over with the variable $\selfLNC$. 
If we are in a context that manipulates a rule, we can also refer to the premises and conclusion with variables $\premisesLNC$ and $\conclLNC$. 
We use the notation $e[v/x]$ to denote the capture-avoiding substitution. 
$\theta$ ranges over finite sequences of substitutions denoted with $[v_1/x_1,\ldots, v_n/x_n]$. $e[v_1/x_2, v_1/x_2, \ldots, v_n/x_n]$ means $((e[v_1/x_1])[v_2/x_2])\ldots[v_n/x_n]$. 
We omit the definition of substitution because it is standard, for the most part. The only aspect that differs from standard substitution is that we do not substitute $\selfLNC$, $\premisesLNC$ and $\conclLNC$ in those contexts that will be set at run-time ($;_{\textsf{r}}$, and selector body). For example, $(\ruleComposition{e_1}{e_2})[v/\mathcal{X}]  \equiv \ruleComposition{(e_1[v/\mathcal{X}])}{e_2}$, where $\mathcal{X}\in \{\selfLNC,\premisesLNC, \conclLNC\}$. 


\subsection{Operational Semantics of $\calcName$}\label{operational}

\begin{figure}[tbp]
\small


\textsf{Dynamic Semantics}  \hfill \fbox{$V ; \LangDef  ;   e\step V ; \LangDef  ;   e$}

\begin{gather*}
	\inference{\{cname\app \langVar{X} ::= v\} \in G}{V ; (G,R)  ;  cname \step\atomic  V ; (G,R)  ;  v} 	\label{beta}\tagsc{r-cname-ok}\\[1ex]
	\inference{\{cname\app \langVar{X} ::= v\} \not\in G}{V ; (G,R)  ;  cname \step\atomic  V ; (G,R)  ;  \errorLNC} 	\label{beta}\tagsc{r-cname-fail}\\[1ex]
	V ; (G,R)  ;  \getRulesLNC \step\atomic  V ; (G,R)  ;  R 	\label{beta}\tagsc{r-getRules}\\[1ex]
	V ; (G,R)  ;  \setRulesLNC{v} \step\atomic  V ; (G,v)  ;  \skipLNC 	\label{beta}\tagsc{r-setRules}\\[1ex]
	\inference{G' = (G \backslash cname) \cup \{cname\app \langVar{X} ::= v\}}{V ; (G,R)  ;  (cname\app \langVar{X} ::= v)   \step\atomic  \emptyset ; (G',R)  ;  \skipLNC } 	\label{beta}\tagsc{r-new-syntax}\\[1ex]
	\inference{\{cname\app \langVar{X} ::= v'\} \in G}{V ; (G,R)  ;  (cname \app \langVar{X} ::= \ldots \app v)   \step\atomic  \emptyset ; (G,R)  ;  cname\app \langVar{X} ::= v' @ v } 	\label{beta}\tagsc{r-add-syntax-ok}\\[1ex]
	\inference{\{cname\app \langVar{X} ::= v'\} \not\in G }{V ; (G,R)  ;  (cname \app \langVar{X} ::= \ldots \app v)  \step\atomic  \emptyset ; (G',R)  ;  \errorLNC} 	\label{beta}\tagsc{r-add-syntax-fail}\\[1ex]
	V ; \LangDef  ;  (\skipLNC ; e) \step\atomic  V ; \LangDef  ;  e	\label{beta}\tagsc{r-seq}\\[1ex]
	V ; \LangDef  ;  \ruleComposition{v}{e}  \step\atomic  V ; \LangDef  ;  e\subsForRule{v} 	\label{r-rule-comp}\tagsc{r-rule-comp}\\[1ex]
	V ; \LangDef  ;  \select{\emptyListLNC}{p}{e} \step\atomic  V ; \LangDef  ;  \emptyListLNC 	\label{beta}\tagsc{r-selector-nil}\\[2ex]
	\inference{\matchLNC{v_1}{p} = \theta \qquad
	\theta' = \mbox{$\begin{cases}
	                              \subsForRule{r}
	                              & \mbox{if } v_1=r\\ 
	                          \{\selfLNC \mapsto v_1\} & \mbox{otherwise}
	                        \end{cases}
	                     $}}
	                     {V ; \LangDef  ;  \select{(\consLNC{v_1}{v_2})}{p}{e} \step\atomic  V ; \LangDef  ;  
	                     (\consLNCFilterName\app {e\theta\theta'}\app {(\select{v_2}{p}{e})})} 	
	                     \label{beta}\tagsc{r-selector-cons-ok}\\[1ex]
	\inference{\matchLNC{v_1}{p} \not= \theta}  
	                     {V ; \LangDef  ;  \select{(\consLNC{v_1}{v_2})}{p}{e} \step\atomic  V ; \LangDef  ;  
	                     (\select{v_2}{p}{e})} 	
	                     \label{beta}\tagsc{r-selector-cons-fail}\\[1ex]
	\inference{
	\langVar{X'} \not\in V \cup \varsSEM{\LangDef} \cup \rangeSEM{\tickSEMna}}{V ; (G,R)  ;  \newVar\app \step\atomic  V \cup \{\langVar{X'}\} ; \LangDef  ;  \langVar{X'}} 	\label{r-newar}\tagsc{r-newvar}\\[1ex]
	\inference{(\mathit{lf}', v_2) = \uniquefySEM\subListFormula{\mathit{lf}}{v_1}{str}{\createMapLNC{[]}{[]}}}{V ; \LangDef  ;  \uniquefy{\mathit{lf}}{v_1}{str}{x}{y}{e} \step\atomic  V ; \LangDef  ;  e[\mathit{lf}'/x,v_2/y] }	\label{r-uniquefy-ok}\tagsc{r-uniquefy-ok}\\[1ex]
	\inference{\uniquefySEM\subListFormula{\mathit{lf}}{v_1}{str}{\createMapLNC{[]}{[]}} = \mathit{fail}}{V ; \LangDef  ;  \uniquefy{\mathit{lf}}{v_1}{str}{x}{y}{e} \step\atomic  V ; \LangDef  ; \errorLNC }	\label{r-uniquefy-fail}\tagsc{r-uniquefy-fail}\\[1ex]
\text{where } \subsForRule{r} \equiv [r/\selfLNC,v_1/ \premisesLNC  , v_2/\conclLNC  ]  \qquad \text{if } r=\inference{v_1}{v_2}
\end{gather*}
\caption{Reduction Semantics of $\calcName$}
\label{fig:dynamicsemantics}
\end{figure}

In this section we show a small-step operational semantics for $\calcName$. 
A configuration is denoted with $V ; \LangDef ;   e$, where $e$ is an expression, $\LangDef$ is the language subject of the transformation, and $V$ is the set of meta-variables that have been generated by $\newVar$. Calls to $\newVar$ make sure not to produce name clashes. 


The main reduction relation is $V ; \LangDef  ;   e\step V' ; \LangDef'  ;   e'$, defined as follows. 
Evaluation contexts $E$ are straightforward and can be found in Appendix \ref{evaluationcontexts}. 
{\footnotesize
\begin{gather*}
\inference
	{V ; \LangDef  ;  e\step\atomic V' ; \LangDef'  ;  e' \\ \vdash \LangDef'  }
	{V ; \LangDef  ;  E[e] \step V' ; \LangDef'  ;  E[e']}
\qquad
\inference
	{V ; \LangDef  ;  e\step\atomic V' ; \LangDef'  ;  e' \\ \not\vdash \LangDef' }
	{V ; \LangDef  ;  E[e] \step V ; \LangDef  ;  \errorLNC}
\\[1ex]
\inference
{}	{V ; \LangDef  ;  E[\errorLNC] \step V ; \LangDef  ;  \errorLNC}
\end{gather*}
}
This relation relies on a step $V ; \LangDef  ;  e\step\atomic V' ; \LangDef'  ;  e'$, which concretely performs the step. 
Since a transformation may insert ill-formed elements such as $\typeOf T \app T$ or $\to \app e \app e$ in the language, we also rely on a notion of type checking for language definitions $\vdash \LangDef'$ decided by the language designer. For example, our implementation of $\calcName$ compiles languages to $\lambda$-prolog and detects ill-formed languages at each step, but the logic of Coq, Agda, Isabelle could be used as well. 
Our type soundness theorem works regardless of the definition of $\typeOf \LangDef'$. 



Fig. \ref{fig:dynamicsemantics} shows the reduction relation $V ; \LangDef  ; e\step\atomic V' ; \LangDef'  ;  e'$. 
We show the most relevant rules. The rest of the rules can be found in Appendix \ref{app:operational}.
\\
\indent 
\textsc{(r-cname-ok)} and \textsc{(r-cname-fail)} handle the encounter of a category name. We retrieve the corresponding list of terms from the grammar or throw an error if the production does not exist. 
\\
\indent 
\textsc{(r-getRules)} retrieves the list of rules of the current language, and \textsc{(r-setRules)} updates this list. 
\\
\indent 
\textsc{(r-new-syntax)} replaces the grammar with a new one that contains the new production. The meta-operation $G \backslash cname$ in that rule removes the production with category name $cname$ from $G$ (definition is straightforward and omitted). The position of $cname$ in $(cname\app \langVar{X} ::= v)$ is not an evaluation context, therefore \textsc{(r-cname-ok)} will not replace that name. 
\textsc{(r-add-syntax-ok)} takes a step to the instruction for adding \emph{new} syntax. The production to be added includes both old and new grammar terms. 
\textsc{(r-add-syntax-fail)} throws an error when the category name does not exist in the grammar, or the meta-variable does not match. 
\\
\indent 
\textsc{(r-rule-seq)} applies when the first expression has evaluated, and starts the evaluation of the second expression. (Evaluation context $E ; e$ evaluates the first expression) 
\\
\indent 
\textsc{(r-rule-comp)} applies when the first expression has evaluated to a rule, and starts the evaluation of the second expression where $\subsForRule{v}$ sets this rule as the current rule. 
\\
\indent 
Rules \textsc{(r-selector-*)} define the behavior of a selector.
\textsc{(r-selector-cons-ok)} and \textsc{(r-selector-cons-fail)} make use of the meta-operation $\matchLNC{v_1}{p} = \theta$. If this operation succeeds it returns the substitutions $\theta$ with the associations computed during pattern-matching. The definition of $\mathit{match}$ is standard and is omitted. 
The body is evaluated with these substitutions and with $\selfLNC$ instantiated with the element selected. If the element selected is a rule, then the body is instantiated with $\subsForRule{v}$ to refer to that rule as the current rule. 
The body of the selector always returns an option type. However, $\consLNCFilterName$ is defined as: $\consLNCFilterName \app {e_1}\app {e_2} \equiv \consLNCFilter{e_1}{e_2}$.
Therefore, $\nothingLNC$s are discarded, and values wrapped in \key{just}s are unwrapped. 
\\
\indent 
\textsc{(r-newvar)} returns a new meta-variable and augments $V$ with it. Meta-variables are chosen among those that are not in the language, have not previously been generated by $\newVar$, and are not in the range of $\tickSEMna$. This meta-operation is used by the tick operator to give a prime to meta-variables. 
\textsc{r-newvar} avoids clashes with these variables, too. 
%
%
\\
\indent 
\textsc{(r-uniquefy-ok)} and \textsc{(r-uniquefy-fail)} define the semantics for \key{uniquefy}. They rely on the meta-operation $\uniquefySEM\subRule{\mathit{lf}}{v}{str}{\createMapLNC{[]}{[]}}$, which takes the list of formulae $\mathit{lf}$, the map $v$, the string $str$, and an empty map to start computing the result map. 
%
The definition of $\mathit{uniquefy}\subRule$ is mostly a recursive traversal of list of formuale and terms, and we omit that. It can be found in Appendix \ref{uniquefy}. 
This function can succeed and return a pair $(\mathit{lf}',v_2)$ where $\mathit{lf}'$ is the modified list of formulae and $v_2$ maps meta-variables to the new meta-variables that have replaced it. $\mathit{uniquefy}\subRule$ can also fail. This may happen when, for example, a map such as $\{\to \app \mapsto ``contra"\}$ is passed when $\to$ requires two arguments. 

\subsection{Type System of $\calcName$}\label{typesystem}
\begin{figure}[tbp]
\textsf{Type System (Configurations)}  \hfill \fbox{$\Gamma \typeOf \app V ; \LangDef ; e $}
\small
\begin{gather*}
\inference{ V \cap \varsSEM{\LangDef} = \emptyset  & \typeOf  \LangDef & \emptyset \typeOf  e : \LanguageType}
{
\typeOf \app V ; \LangDef ; e 
}
\end{gather*}
\textsf{Type System (Expressions)}  \hfill  \fbox{$\Gamma \typeOf \app e : T $}
\begin{gather*}
\ninference{t-var}
{}{
\Gamma, x:T\typeOf \app x: T}
\quad
\ninference{t-opname}
	{\Gamma \typeOf \app e : \List\app\TermType } 
	{ \Gamma \typeOf \app (opname\app e) : \TermType}
\quad
\ninference{t-opname-var}
	{ \Gamma \typeOf \app e : \List\app\TermType } 
	{ \Gamma, x : \OpnameType \typeOf \app (x\app e) : \TermType}
\\[2ex]
\ninference{t-meta-var}
{}
	{ \langVar{X}: \TermType}
\qquad
\ninference{t-abs}
	{\Gamma \typeOf \app e : \TermType } 
	{ \Gamma \typeOf \app (\langVar{z})e: \TermType}
\qquad
\ninference{t-subs}
	{\Gamma \typeOf \app e_1 : \TermType \quad \Gamma \typeOf \app e_2 : \TermType } 
	{ \Gamma \typeOf \app e_1[e_2/\langVar{z}]: \TermType}
\\[2ex]
\ninference{t-predname}
	{ \Gamma \typeOf \app e : \List\app\TermType } 
	{ \Gamma \typeOf \app ( predname \app e) : \FormulaType}
\qquad
\ninference{t-predname-var}
	{\Gamma \typeOf \app e : \List\app\TermType } 
	{ \Gamma,x: \PrednameType \typeOf \app ( x \app e) : \FormulaType}
\\[2ex]
\ninference{t-rule}
	{\Gamma \typeOf \app e_1 : \List\app\FormulaType \\\\ \Gamma \typeOf \app e_2 : \FormulaType} 
	{ \Gamma \typeOf \app \inference{e_1}{e_2} : \RuleType}
\qquad
\ninference{t-seq}
	{\Gamma \typeOf \app e_1 : \LanguageType \\\\ \Gamma \typeOf \app e_2 : \LanguageType } 
	{ \Gamma \typeOf \app e_1 ; e_2 : \LanguageType}
\qquad
\ninference{t-rule-comp}
	{ \Gamma \typeOf \app e_1 : \RuleType \\\\ \Gamma,\GammaForRule{r} \typeOf \app e_2 : \RuleType } 
	{ \Gamma \typeOf \app \ruleComposition{e_1}{ e_2} : \RuleType}
\\[1ex]
\ninference{t-selector}
	{\Gamma \typeOf \app e_1 : \List \app T \quad \Gamma \typeOfPattern \app p  : T \Rightarrow \Gamma' \\
	\Gamma'' = \mbox{$ 
				\begin{cases} 
					\GammaForRule{r} & \mbox{if } T = \RuleType \\
					\selfLNC : T &  \mbox{otherwise}
				\end{cases}
				$}\\
				\Gamma, \Gamma', \Gamma'' \typeOf \app e_2 : \MaybeType{T'}} 
	{ \Gamma \typeOf \app \select{e_1}{p}{e_2} : \List \app T'}
\\[2ex]
\ninference{t-syntax-new \text{and} t-syntax-add}
	{\Gamma \typeOf \app e : \List \app \TermType} 
	{\mbox{$\begin{array}{c}
	  \Gamma \typeOf \app cname \app \langVar{X} ::= e : \LanguageType\\
	  \Gamma \typeOf \app cname \app \langVar{X}::= \ldots \app e: \LanguageType
	  \end{array}$
	  }}
%
\qquad 
\ninference{t-cname}
{}
{\Gamma \typeOf cname : \List \app \TermType}
\\[2ex]
\ninference{t-getRules}
{}
{\Gamma \typeOf \getRulesLNC : \List \app \RuleType}
\qquad 
\ninference{t-setRules}
	{\Gamma \typeOf \app e : \List \app \RuleType}
	{\Gamma \typeOf \app \setRulesLNC{e} : \LanguageType}
\\[2ex]
\ninference{t-uniquefy}
	{\Gamma \typeOf \app e_1 : \List\app\FormulaType \\\\
	\Gamma \typeOf \app e_2 : \MapType{T'}{(\List \app \StringType)} \quad T' = \OpnameType \text{ or } T' = \PrednameType \\\\
	\Gamma, x: \List\app\FormulaType,  y: \MapType{\TermType}{(\List\app \TermType)} \typeOf \app e_3 : T 
	} 
	{ \Gamma \typeOf \app \uniquefy{e_1}{e_2}{str}{x}{y}{e_3} : T}
\qquad 
\begin{array}{c}
\ninference{t-skip}
{}
{\Gamma \typeOf \app \skipLNC: \LanguageType}
\\[2ex]
\ninference{t-newvar}
{}
	{ \Gamma \typeOf \app \newVar : \TermType}
\end{array}\\[2ex]
\text{where } \GammaForRule{r} \equiv \selfLNC : \RuleType, \premisesLNC : \List \app \FormulaType,  \conclLNC: \FormulaType
\end{gather*}
\caption{Type System of $\calcName$}
\label{fig:typesystem}
\end{figure}

In this section we define a type system for $\calcName$. 
Types are defined as follows
{\small
\begin{syntax}
  \text{\sf Type} & T & ::= &  \LanguageType  \mid  \RuleType \mid \FormulaType \mid \TermType  \\ 
  &&& \List\app T \mid \MapType{T}{T} \mid \MaybeType{T}\mid \StringType \mid \OpnameType \mid \PrednameType\\
   \text{\sf Type Env} & \Gamma & ::= &  \emptyset \mid \Gamma, x:T 
\end{syntax}
}

\noindent We have a typical type environment that maps variables to types. 
Fig. \ref{fig:typesystem} shows the type system. 
The typing judgement $\typeOf  V ; \LangDef ; e$ means that the configuration $V ; \LangDef ; e$ is well-typed. 
This judgment checks that the variables of $V$ and those in $\LangDef$ are disjoint. 
This is an invariant that ensures that $\newVar$ always produces fresh names. 
We also check that $\LangDef$ is well-typed and that $e$ is of type $\LanguageType$.

We type check expressions with the typing judgement $\Gamma \typeOf \app e : T$, which means that $e$ has type $T$ under the assignments in $\Gamma$. 
%
Most typing rules are straightforward. We omit rules about lists and maps because they are standard. We comment only on the rules that are more involved. \textsc{(t-selector)} type checks a selector operation. We use $\Gamma \typeOfPattern \app p  : T \Rightarrow \Gamma'$ to type check the pattern $p$ and return the type environment for the variables of the pattern. Its definition is standard and is omitted. 
When we type check the body $e_2$ we then include $\Gamma'$. If the elements of the list are rules then we also include $\GammaForRule{r}$ to give a type to the variables for referring to the current rule. Otherwise, we assign $\selfLNC$ the type of the element of the list. Selectors with \key{keep} are analogous and omitted. 
\textsc{(t-rule-comp)} type checks a rule composition. In doing so, we type check the second expression with $\GammaForRule{r}$. 
\textsc{(t-uniquefy)} type checks the \key{uniquefy} operation. As we rename variables depending on the position they hold in terms and formulae, the keys of the map are of type $\OpnameType$ or $\PrednameType$, and values are strings. We type check $e_3$ giving $x$ the type of list of formulae, and $y$ the type of a map from meta-variables to list of meta-variables. 

We have proved that $\calcName$ is type sound. 

\begin{theorem}[Type Soundness]
For all $\Gamma$, $V$, $\LangDef$, $e$, if $\typeOf V ; \LangDef ; e$ then $V ; \LangDef ; e \step^{*} V' ; \LangDef' ; e'$ s.t. i) $e' = \skipLNC$, ii) $e' = \errorLNC$, or iii) $V' ; \LangDef' ; e' \step V'' ; \LangDef'' ; e''$, for some $e''$.
\end{theorem}

The proof is by induction on the derivation $\typeOf V ; \LangDef ; e$, and follows the standard approach of Wright and Felleisen \cite{WrightFelleisen94} through a progress theorem and a subject reduction theorem. 
%
The proof can be found in Appendix \ref{proof}.

\section{Examples}\label{examples}

We show the applicability of $\calcName$ with two examples of language transformations: adding subytyping \cite{tapl} and switching to big-step semantics \cite{Kahn87}. 
 In the code we use let-binding, pattern-matching, and an overlap operation that returns true if two terms have variables in common. These operations can be easily defined in $\calcName$, and we show them in Appendix \ref{let}. 
The code below defines the transformation for adding subtyping. We assume that two maps are already defined, $mode = \{\typeOf\app\mapsto[``inp",\app``inp",\app``out"]\}$ and $variance = \{\to\app \app\mapsto[``contra",\app ``cova"]\}$. 
\lstset
{
    numbers=left,
    stepnumber=1,
}
\begin{lstlisting}[mathescape=true]
$\key{setRules}$
 $\key{getRules}(\key{keep})[(\typeOf [\Gamma, e, T])]:$
 $\key{uniquefy}(premises, mode, ``out") => (uniq, newpremises):$
  $\underline{newpremises \app @ \app \key{concat}(\key{mapKeys}(uniq)[T_f]: \key{fold} <: uniq(T_f))}$
                 $conclusion$
                     $;_{\textsf{r}}$
  $\key{concat}(premises(\key{keep})[T_1 <: T_2]:$
   $premises[(\typeOf [\Gamma, e_v, (c_v \app Ts_v)])]:$
    $\key{let}\app vmap = \key{map}(Ts_v,\app variance(c_v)) \app\key{in}\app$
    $\key{if}\app vmap(T_1) = ``contra"  \app\key{then}\app T_2 <: T_1$
    $\underline{\key{else}\app\key{if}\app vmap(T_1) = ``inv" \app\key{and}\app vmap(T_2) = ``inv" \app\key{then}\app T_1 = T_2 \app\key{else}\app T_1 <: T_2)}$                                                                     $conclusion$
                     $;_{\textsf{r}}$
   $\key{let}\app \mathit{outputVars} = \key{match}\app conclusion \app\key{with}\app (\typeOf [\Gamma, e_c, T_c]) \Rightarrow \varsLNC{T_c} \app\key{in}\app$
   $\key{let}\app joins = \key{mapKeys}(uniq)[T_i]:$
     $\key{if}\app T_i \app\key{in}\app \mathit{outputVars} \app\key{then}\app (\sqcup\app uniq(T_i)\app =\app T_i) \app\key{else}\app \nothingLNC$
   $\key{in}\app\underline{premises  \app @ \app joins}$
        $conclusion$
\end{lstlisting}
Line 1 updates the rules of the language with the rules computed by the code in lines 2-17. 
Line 2 selects all typing rules, and each of them will be the subject of the transformations in lines 3-17.  
Line 3 calls $\key{uniquefy}$ on the premises of the selected rule. We instruct $\key{uniquefy}$ to give new variables to the outputs of the typing relation $\typeOf$, if they are used  more than once in that position. 
As previously described, $\key{uniquefy}$ returns the list of new premises, which we bind to $\mathit{newpremises}$, and the map that assigns variables to the list of the new variables generated to replace them, which we bind  to $uniq$. 
The body of $\key{uniquefy}$ goes from line 4 to 17. 
Lines 4 and 5 build a new rule with the conclusion of the selected rule (line 5). It does so using the special variable name \emph{conclusion}. The premises of this rule include the premises just generated by $\key{uniquefy}$. Furthermore, we add premises computed as follows. With $\key{mapKeys}(uniq)[T_f]$, we iterate over all the variables replaced by $\key{uniquefy}$. We take the variables that replaced them and use fold to relate them all with subtyping. In other words, for each $\{T\app\mapsto[T_1, \ldots, T_n]\}$ in $uniq$, we have the formulae $T_1 <: T_2, \ldots, T_{n-1} <: T_n$.
This transformation has created a rule with unique outputs and subtyping, but subtyping may be incorrect because if some variable is contravariant its corresponding subtyping premise should be swapped. 
Lines 7-11, then, adjust the subtyping premises based on the variance of types. Line 7 selects all subtyping premises of the form $T_1 <: T_2$. For each, Line 8 selects typing premises with output of the form $(c_v \app Ts_v)$. We do so to understand the variance of variables. If the first argument of $c_v$ is contravariant, for example, then the first element of $Ts_v$ warrants a swap in a subtyping premise because it is used in the contravariant position. 
We achieve this by creating a map that associates the variance to each argument of $c_v$. The information about the variance for $c_v$ is in $variance$.  
If $T_1$ or $T_2$ (from the pattern of the selected premise) appear in $Ts_v$ then they find themselves with a variance assigned in $vmap$. 
Lines 10-11 generate a new premise based on the variance of variables. For example, if $T_1$ is contravariant then we generate $T_2 <: T_1$.

The program written so far (lines 1-11) is enough to add subtyping to several typing rules. For example, \textsc{(t-app)} can be transformed into \textsc{(t-app')} with this program. 
However, some typing rules need a more sophisticated algorithm. Below is the typing rule for if-then-else on the left, and its version with subtyping on the right, which makes use of the join operator ($\sqcup$) (see, \cite{tapl}). 
\[
{\footnotesize 
\begin{array}{ccc}
{
\inference
	{
	\Gamma \typeOf \app e_1 : \Bool \\
	\Gamma \typeOf \app e_2 : T	\quad
	\Gamma \typeOf \app e_2 : T	
	} 
	{ \Gamma \typeOf \app (\mathit{if} \app e_1\app e_2\app e_3) : T}
}
& ~~ \Longrightarrow  ~~ &
{
\inference
  {\Gamma \typeOf e_1 : \Bool \qquad \Gamma \typeOf e_2 : T_1 \\\\ \Gamma \typeOf e_3 : T_2 \qquad T_1 \sqcup T_2 = T
  }
  {\Gamma \typeOf (\mathit{if} \app e_1 \app e_2 \app e_3) : T}
}
\end{array}
}
\]
If we removed $T_1 \sqcup T_2$ the meta-variable $T$ would have no precise instantiation because its counterpart variables have been given new names. 
Lines 13-17 accommodate for cases of the like. Line 13 saves the variables that appear the output type of the rule in \emph{outputVar}. We then iterate over all the keys of $uniq$, that is, the variables that have been replaced. For each of them, we see if they appear in \emph{outputVar}. If so then we create a join operator with the variables newly generated to replace this variable, which can be retrieved from $uniq$. We set the output of the join operator to be the variable itself, because that is the one used in the conclusion. 

The algorithm above shows that \key{uniquefy} is a powerful operation of $\calcName$. 
To illustrate \key{uniquefy} further, let us consider a small example before we address big-step semantics. 
Suppose that we would like to make every test of equality explicit. We therefore want to disallow terms such as $(\key{op}\app e\app e\app e)$ to appear in the premises, and want to turn them into $(\key{op}\app e_1\app e_2\app e_3)$ together with premises $e_1 = e_2$ and $e_2 = e_3$. In $\calcName$ we can do this in the following way. Below we assume that the map \emph{allOps} maps each operator to the string ``yes" for each of its arguments. This instructs \key{uniquefy} to look for every argument. 

\lstset
{
    numbers=left,
    stepnumber=1,
}
\begin{lstlisting}[mathescape=true]
...
 $\key{uniquefy}(premises, allOps, ``yes") => (uniq, newpremises):$
  ${newpremises \app @ \app \key{concat}(\key{mapKeys}(uniq)[T_f]: \key{fold} = uniq(T_f))}$
\end{lstlisting}
Below, we show the code to turn language definitions into big-step semantics.
\lstset
{
    numbers=left,
    stepnumber=1,
}
\begin{lstlisting}[mathescape=true]
$\key{setRules}$
 $\mathit{Value}[v]: v \step v\app @$
 $\key{getRules}(\key{keep})[(op \app es) \app\step\app et]:$
 $\key{if}\app \key{isEmpty}(\select{Expression}{(op\app \_)}{self})\app\key{then}\app \key{nothing} \app\key{else}\app$
 $\key{let}\app v_{res} = \key{newVar}\app\key{in}$
 $\key{let}\app emap = \key{createMap}((\select{es}{e}{\newVar}), es) \app\key{in}$
  $(\key{mapKeys}(emap)[e]: \key{if}\app \key{isVar}(emap(e))\app \key{and}\app \key{not}(emap(e)\app\key{in}\app \varsLNC{et})$
                 $\key{then}\app \key{nothing}  \app\key{else}\app e \step emap(e) )$
  $\underline{@\app (\key{if}\app \key{not}(et \app\key{in}\app es)\app\key{then}\app [(et \app\step\app v_{res})] \app\key{else} \app \key{nil}) \app @ \app premises\qquad\qquad\qquad ~~~}$
  ${(op \app (\key{mapKeys}(emap))) \app\step\app \key{if}\app\key{not}(et \app\key{in}\app es)\app\key{then}\app v_{res} \app\key{else}\app et}$
\end{lstlisting}

\lstset
{
    numbers=left,
    stepnumber=1,
    basicstyle=\footnotesize,
}

%
Line 1 updates the rules of the language with the list computed in lines 2-9. Line 2 generates reduction rules such as $\lambda x.e \step \lambda x. e$, for each value, as it is standard in big-step semantics. These rules are appended to those generated in lines 3-9. Line 3 selects all the reduction rules. Line 4 leaves out those rules that are not about a top-level expression operator. This skips contextual rules that take a step $E[e] \step E[e']$, which do not appear in big-step semantics. 
To do so, line 4 make use of $\select{\emph{Expression}}{(op\app \_)}{\emph{self}})$. 
As $op$ is bound to the operator we are focusing on (from line 2), this selector returns a list with one element if $op$ appears in \emph{Expression}, and an empty list otherwise. 
This is the check we perform at line 4. 
Line 5 generates a new variable that will store the final value of the step. 
Line 6 assigns a new variable to each of the arguments in $(es)$. We do so creating a map \emph{emap}. 
These new variables are the formal arguments of the new rule being generated (Line 9). 
Line 7-8 makes each of these variables evaluate to its corresponding argument in $es$ (line 8). For example, for the beta-reduction an argument of $es$ would be $\lambda x. e$ and we therefore generate the premise $e_1 \step \lambda x. e$, where $e_1$ is the new variable that we assigned to this argument with line 6. 
Line 7 skips generating the reduction premise if it is a variable that does not appear in $e_t$. For example, in the translation of \textsc{(if-true)} $(\mathit{if} \app true \app e_2 \app e_3) \step e_2$ we do not evaluate $e_3$ at all. 
Line 9 handles the result of the overall small-step reduction. This result is evaluated to a value ($v_{res}$), unless it already appears in the arguments $es$. 
The conclusion of the rule syncs with this, and we place $v_{res}$ or $e_t$ in the target of the step accordingly. 
Line 9 also appends the premises from the original rule, as they contain conditions to be checked. 

When we apply this algorithm to the simply typed $\lambda$-calculus with if-then-else 
 we obtain: (we use standard notation rather than $\calcName$ syntax)
{\footnotesize
\begin{align*}
%
(\lambda x.e \app v) \step e[v/x]
\qquad&\Rightarrow\qquad
\inference
  {e_1' \step \lambda x.e & e_2' \step v& e[v/x]\step v_{res}}
   {(e_1'\app e_2')\step v_{res}}
\\[2ex]
(\mathit{if} \app true \app e_1 \app e_2) \step e_1
\qquad&\Rightarrow\qquad
\inference 
  {e_1' \step {true} 
  & e_2' \step e_2}
  {(\mathit{if} \app e_1' \app e_2' \app e_3') \step e_2}
\end{align*}
}

We have implemented $\calcName$ and we have applied it to the examples in this paper as well as $\lambda$-calculi with lists, pairs, sums, options, let-binding, function composition $(g\circ f)(x)$, and System F. We also considered these calculi in both call-by-value and call-by-name version, as well as lazy evaluation for data types such as pairs and lists. 
The languages produced by our tool are compiled to $\lambda$-prolog, which type checks them successfully and, in fact, can execute them. We have tested subtyping with simple programs and checked that this functionality has been added. 
We have also tested big-step evaluations with simple programs and our tests evaluate to the expected values in one step. 

The tool, repo of languages generated, and details of our tests can be found at the website of our tool \cite{ltr}. 

\section{Related Work}\label{related}

An excellent classification of language transformations has been provided in \cite{ErdwegGR12}. 
The paper defines five operations: language extension, language unification, language restriction, self-extension of an (embedded) language, and support for composition of extensions.
Language workbenches (Rascal, Spoofax, etcetera) implement these types of transformations and similar ones. 
These transformations are coarse grained in nature because they do not access the components of languages with precision. $\calcName$, instead, includes operations to scan rules, and select/manipulate formulae and terms with precision. In this regard, we offer low-level manipulations, and yet programmers can enjoy a rather declarative language. We are not aware of calculi that provide these features. We are also not aware of type soundness proofs of similar calculi.  

Proof assistants are optimized for handling inductive (rule-based) definitions, and can automatically generate powerful inductive reasoning mechanisms from these definitions. 
$\calcName$ does not provide these features, and does not assist language designers with their proofs. On the other hand, proof assistants do not have reflective features for programmatically retrieving their own inductive definitions, selected by a pattern, and for manipulating them to form a different specification, which is instead characteristic of $\calcName$. It would be interesting to explore the merging of the features of proof assistants and $\calcName$ in one tool. 
Another limitation of $\calcName$ compared to proof assistants (and general-purpose programming languages) is that $\calcName$ does not offer recursion but only a simple form of iteration. 
%

We are not aware of algorithms that automatically add subtyping. Instead, there has been some work on deriving big-step semantics, which we discuss below. 
The work of Danvy et al \cite{danvy2004refocusing,Danvy:2008,danvy:reduction-free} offers a comprehensive translation from small-step to big-step semantics. The approach derives small-step abstract machines first, which are then translated into to big-step abstract machines and finally into big-step reduction semantics. The approach is rather elaborate and involves techniques such as refocusing and transition compression. It would be interesting to express these algorithms in $\calcName$, extending the calculus if needed. 
Our algorithm differs slightly from that of Ciob\^ac\u{a} \cite{Ciobaca13}. 
Rules in \cite{Ciobaca13} have been proved correct. 
We do not offer correctness theorems about our algorithms, and we have shown them solely to demonstrate the kinds of manipulations that our calculus offers.

\section{Conclusion}\label{conclusion}

We have presented $\calcName$, a calculus for expressing language transformations. The calculus is expressive enough to model interesting transformations such as adding subtyping, and switching from small-step to big-step semantics. We have proved the type soundness of $\calcName$, and we have implemented the calculus in a tool. 
As $\calcName$ manipulates inference systems it can, in principle, be applied to logical systems, and we plan to explore this research venue. 
Overall, we believe that the calculus offers a rather declarative style for manipulating languages. 


\bibliographystyle{splncs04}
\bibliography{all.bib} 

\appendix

\section{Evaluation Contexts}\label{evaluationcontexts}

\begin{syntax}
   \text{\sf Evaluation Contexts} & E & ::= & \consLNC{E}{e}\mid \consLNC{v}{E}\mid \head \app E \mid \tail \app E \mid E @ e \mid v @ E \\
   &&& \mid \createMapLNC{E}{e} \mid \createMapLNC{v}{E} \mid  E(e)\mid  v(E)\mid  \mapKeysLNC{E}\\
   &&& \mid \justLNC{E} \mid \getLNC{E}\\
   &&& \mid cname \app \langVar{X} ::= E  \app  \mid cname \app \langVar{X} ::= \ldots \app E \\
   &&&\mid  \setRulesLNC{E}\\
   &&& \mid \select{E}{p}{e} \mid \select{E(\key{keep})}{p}{e} \\
    &&& \mid \uniquefy{E}{e}{str}{x}{x}{e}\\
    &&& \mid \uniquefy{v}{E}{str}{x}{x}{e} \\
    &&&\mid \foldLNC{predname}{E}\\
   &&& \mid  E \app ; \app e \mid \ruleComposition{E}{e}   \\
   &&&  \mid \varsLNC{E}  \mid \tickedMMLNC{E}{e} \\ 
   &&&\mid \ifLang{E == e}{e}{e}\mid \ifLang{v ==E}{e}{e}\\
   &&&\mid \ifLang{ \isEmptyLNC{E}}{e}{e}\\
   &&&\mid \ifLang{ \isIn{E}{e}}{e}{e}\mid \ifLang{ \isIn{v}{E}}{e}{e}\\
   &&&\mid \ifLang{ \isNothing{E}}{e}{e}\\
   &&& \mid \inference{e}{E} \mid \inference{E}{v}\\
   &&&  \mid predname \app E\\
   &&& \mid  opname\app E \mid  (\langVar{X})E  \mid E[e/\langVar{X}] \mid v[E/\langVar{X}] 
\end{syntax}

\newpage

\section{Rest of Operational Semantics}\label{app:operational}

We provide the reduction rules for the selector with keep, fold, and the tick operator. 

\begin{gather*}
	V ; \LangDef \app ; \app \select{\emptyListLNC\key{(keep)}}{p}{e} \step\atomic  V ; \LangDef \app ; \app \emptyListLNC 	\label{beta}\tagsc{r-selector-nil}\\[2ex]
	\inference{\matchLNC{v_1}{p} = \theta \qquad
	\theta' = \mbox{$\begin{cases}
	                              \subsForRule{r}
	                              & \mbox{if } v_1=r\\ 
	                          \{\selfLNC \mapsto v_1\} & \mbox{otherwise}
	                        \end{cases}
	                     $}}
	                     {V ; \LangDef \app ; \app \select{(\consLNC{v_1}{v_2})\key{(keep)}}{p}{e} \step\atomic  V ; \LangDef \app ; \app 
	                     (\consLNCFilterName\app {e\theta\theta'}\app {(\select{v_2}{p}{e})})} 	
	                     \label{beta}\tagsc{r-selector-cons-ok}\\[1ex]
	\inference{\matchLNC{v_1}{p} \not= \theta}  
	                     {V ; \LangDef \app ; \app \select{(\consLNC{v_1}{v_2})\key{(keep)}}{p}{e} \step\atomic  V ; \LangDef \app ; \app 
	                     \key{cons}\app v_1\app (\select{v_2}{p}{e})} 	
	                     \label{beta}\tagsc{r-selector-cons-fail}\\[1ex]
	\begin{array}{c}
	V ; \LangDef \app ; \app \foldLNC{predname}{(\consLNC{v_1}{(\consLNC{v_2}{v_3})})} \\ \step\atomic  \\ V ; \LangDef \app ; \app (\consLNC{(predname \app [v_1, v_2]}(\foldLNC{predname}{(\consLNC{v_2}{v_3})}))	\label{r-fold-2}\tagsc{r-fold-2}
	\end{array}\\[1ex]
	V ; \LangDef \app ; \app \foldLNC{predname}{(\consLNC{v}{\emptyListLNC})} \step\atomic  V ; \LangDef \app ; \app \emptyListLNC	\label{r-fold-1}\tagsc{r-fold-1}\\[1ex]
	V ; \LangDef \app ; \app \foldLNC{predname}{\emptyListLNC} \step\atomic  V ; \LangDef \app ; \app \emptyListLNC	\label{r-fold-0}\tagsc{r-fold-0}\\[1ex]
		V ; \LangDef \app ; \app \tickedMMLNC{(opname\app v_1)}{v_2} \step\atomic  V ; \LangDef \app ; \app (opname \app \tickedMMLNC{ v_1}{v_2})	\label{r-tick-opname}\tagsc{r-tick-opname}\\
	V ; \LangDef \app ; \app \tickedMMLNC{\langVar{X}}{v_2} \step\atomic  V ; \LangDef \app ; \app \ifLang{\isIn{\langVar{X}}{\key{vars}(v_2)}}{\tickSEM{\langVar{X}}}{\langVar{X}} 	\label{r-tick-var}\tagsc{r-tick-var}\\
	V ; \LangDef \app ; \app \tickedMMLNC{((\langVar{z})t)}{v_2} \step\atomic  V ; \LangDef \app ; \app (\langVar{z})(\tickedMMLNC{t}{v_2})	\label{r-tick-abs}\tagsc{r-tick-abs}\\
	V ; \LangDef \app ; \app \tickedMMLNC{(t_1[t_2/\langVar{z})]}{v_2} \step\atomic  V ; \LangDef \app ; \app (\tickedMMLNC{t_1}{v_2})[(\tickedMMLNC{t_2}{v_2})/\langVar{z}]	\label{r-tick-sub}\tagsc{r-tick-sub}\\
	V ; \LangDef \app ; \app \tickedMMLNC{(\consLNC{v_1}{v_2})}{v_3} \step\atomic  V ; \LangDef \app ; \app (\consLNC{(\tickedMMLNC{v_1}{v_3})}{(\tickedMMLNC{v_2}{v_3})}) 	\label{r-tick-cons}\tagsc{r-tick-cons}\\
	V ; \LangDef \app ; \app \tickedMMLNC{\emptyListLNC}{v_2} \step\atomic  V ; \LangDef \app ; \app \emptyListLNC \label{r-tick-nil}\tagsc{r-tick-nil}\\
\end{gather*}

The reduction rules for lists, maps, options and if-then-else are standard and are omitted. We omit the reduction rules for vars() because they are straightforward. They traverse terms and list of terms in a similar way as the tick operator, though they expose the variables met along the way. 

\newpage
\section{Uniquefy}\label{uniquefy}

{\small
\begin{gather*}
\begin{array}{l}
\uniquefySEM\subListFormula{\emptyListLNC}{m}{str}{mr} \equiv (\emptyListLNC, mr) \\
\uniquefySEM\subListFormula{(\consLNC{v}{lv})}{m}{str}{mr} \equiv \\
\qquad \letSEM\app (v', mr') = \uniquefySEM\subFormula{v}{m}{str}{mr} \app \inSEM \app \\
\qquad\qquad\qquad \letSEM \app (lv', mr'') \eqSEM \uniquefySEM\subListFormula{lv}{m}{str}{mr'} \app \inSEM \app ((\consLNC{v'}{lv'}), mr'')\\
\uniquefySEM\subFormulaAndTerm{(name \app v)}{m}{str}{mr} \equiv (\text{where } name \text{ is } predname \text{ or } opname)\\
\qquad\ifSEM \app name \app \inSEM \app m \app  \thenSEM \app \ifSEM  \app\checkzipSEM{v}{m.name}  \app\thenSEM \app \\
\qquad\qquad\qquad \letSEM \app (v', mr') \eqSEM  \uniquefySEMfound{(\zipSEM{v}{m.name}}{m}{str}{mr} \app \inSEM \app ((name \app v'), mr') \app \\
\qquad\qquad\qquad\qquad\qquad\elseSEM \app \failureSEM\\
\qquad \elseSEM \app \letSEM \app (v', mr') \eqSEM \uniquefySEM\subListTerm{v}{m}{str}{mr} \app \inSEM \app ((name \app v'), mr')\\
\uniquefySEM\subTerm{X}{m}{str}{mr} \equiv (X, mr)\\    
\uniquefySEM\subTerm{x)t}{m}{str}{mr} \equiv ((x)t , mr)\\
\uniquefySEM\subTerm{t[t/x]}{m}{str}{mr} \equiv ( t[t/x], mr)\\[2ex]
\uniquefySEMfound{\emptyListLNC}{m}{str}{mr} \equiv (\emptyListLNC, mr)  \\
\uniquefySEMfound{(\consLNC{(v,str')}{lv})}{m}{str}{mr} \equiv  \\
\qquad\letSEM \app (lv', mr') \eqSEM  \uniquefySEMfound{lv}{m}{str}{mr} \app \inSEM \app\\
\qquad\qquad\qquad \ifSEM \app str \app \eqSEM\eqSEM str' \thenSEM \app  ((\consLNC{v}{lv'}),mr') \\
\qquad\qquad \qquad\quad\elseSEM \app \letSEM \app (v', mr'') \eqSEM \uniquefySEMreplace\subTerm{v}{mr'} \app \inSEM ((\consLNC{v'}{lv'}),mr'')\\[2ex]
\uniquefySEMreplace\subFormulaAndTerm{(name \app v)}{mr} \equiv (\text{where } name \text{ is } predname \text{ or } opname)\\
\qquad\letSEM \app (v', mr') \eqSEM \uniquefySEMreplace\subTerm{v}{mr} \app \inSEM \app ((name \app v'), mr')\\
\uniquefySEMreplace\subTerm{{X}}{\createMapLNC{keys}{values}\app as \app mr} \equiv \\
\qquad\letSEM \app (keys', values') \eqSEM  remove(X,mr) \app \inSEM \app\\
\qquad\ifSEM \app X \app \inSEM \app mr \app  \thenSEM \app (\tickSEM{\lastSEM{mr.X}}, \createMapLNC{keys' \app @ \app [X]}{values' \app @ \app [\tickSEM{\lastSEM{mr.X}}]})\\ 
\qquad\qquad \qquad\quad\elseSEM (\tickSEM{X}, \createMapLNC{keys' \app @ \app [X]}{values' \app @ \app [\tickSEM{X}]})\\
\uniquefySEMreplace\subTerm{(x)tv}{mr} \equiv \letSEM \app (tv', mr') \eqSEM \uniquefySEMreplace\subTerm{tv}{mr} \app \inSEM \app((x)tv' , mr')\\
\uniquefySEMreplace\subTerm{tv_1[tv_2/x]}{mr} \equiv 
\qquad \letSEM\app (tv_1', mr') \eqSEM \uniquefySEM\subTerm{tv_1}{m}{str}{mr}\app \inSEM \\
\qquad\qquad\qquad \letSEM \app (tv_2', mr'') \eqSEM \uniquefySEM\subTerm{tv_2}{m}{str}{mr'} \app \inSEM \app (tv_1'[tv_2'/x], mr'')
\end{array}
\end{gather*}
}

\noindent $remove(t,m)$ removes the key $t$ and its corresponding value in the map $m$. \\
$last(l)$ returns the last element of a list, and $@$ is list append. 

The function is mostly a straightforward recursive traverse of terms, formulae, list of terms and list of formulae. 
The only elements to notice are that when $\mathit{uniquefy}$ detects a context that potentially contain the string $str$ then it switches to $\mathit{uniquefy}^{\bullet}$, which is a meta-operation that seeks for $str$. In turn, when $\mathit{uniquefy}^{\bullet}$ finds an argument in a position prescribed by $str$, then it switches to $\mathit{uniquefy}^{\dagger}$, which is a meta-operation that is responsible for actually replace variables and record the association. $\mathit{zip}$ is a meta-operation that combines two lists. Of course, it may fail if the two lists do not have the same length. This happens in the scenario described above about $\to$ and its number of argumets. $\mathit{CheckZip}$ performs just that check and can make the function fail. 

\section{Proof of Type Soundness}\label{proof}

\subsection{Progress Theorem}

\begin{theorem}[Canonical Form Lemmas]

\begin{itemize}
\item $\emptyset \typeOf \app e : \LanguageType$, and $e$ is a value then $e = \skipLNC$.
\item $\emptyset \typeOf \app e : \RuleType$, and $e$ is a value then $e = r$.
\item $\emptyset \typeOf \app e : \FormulaType$, and $e$ is a value then $e = f$.
\item $\emptyset \typeOf \app e : \TermType$, and $e$ is a value then $e = t$.
\item $\emptyset \typeOf \app e : \List \app T$, and $e$ is a value then $e = \nil \app \lor e = \cons{\app v_1}{v_2}$.
\item $\emptyset \typeOf \app e : \MapType{T_1}{T_2}$, and $e$ is a value then $e = \createMapLNC{v_1}{v_2}$.
\item $\emptyset \typeOf \app e : \MaybeType{T_1}{T_2}$, and $e$ is a value then $e = \nothingLNC\app \lor  = \justLNC{v}$.
\item $\emptyset \typeOf \app e : \StringType$, and $e$ is a value then $e = str$.
\item $\emptyset \typeOf \app e : \OpnameType$, and $e$ is a value then $e = opname$.
\item $\emptyset \typeOf \app e : \PrednameType$, and $e$ is a value then $e = predname$.
\end{itemize}

\end{theorem}
\begin{proof}
Each case is proved by case analysis on $\emptyset \typeOf \app e : T$. Each case is straightforward. 
\end{proof}

\begin{theorem}[Progress Theorem Expressions]
For all , if $\emptyset \typeOf e : T$ then either 
\begin{itemize}
\item $e = v$, or 
\item $e = \errorLNC$, or 
\item for all $V, \LangDef$, $V ; \LangDef ; e \step V' ; \LangDef' ; e'$, for some $V',\LangDef',e'$.
\end{itemize}
\end{theorem}
\begin{proof}
We prove the theorem by induction on the derivation of $\emptyset \typeOf e: T$. 
Let us assume the proviso of the theorem, that is (H1) $\emptyset \typeOf e: T$. 
\begin{case}[\textsc{t-opname}] 
Since (H1) then we have $e= (opname \app e)$ with $\emptyset \typeOf e: \List\app \TermType$.
By IH, we have that 
\begin{itemize}
\item $e = v$. Then we have $(opname \app v)$ which is a value. 
\item $e = \errorLNC$. Then we have $(opname \app \errorLNC)$ and by \textsc{ctx-err} we have \\$V ; \LangDef \app ; \app (opname \app \errorLNC) \step V ; \LangDef \app ; \app \errorLNC$, for all $V, \LangDef$ because of the evaluation context $(opname \app E)$.
\item for all $V, \LangDef$, $V ; \LangDef ; e \step V' ; \LangDef' ; e'$, for some $V',\LangDef',e'$. Then $(opname \app e)$ takes a step by \textsc{ctx-succ} or \textsc{ctx-lang-err}. 
\end{itemize}
\end{case}
\begin{case}[\textsc{t-rule-comp}] 
Since (H1) then we have $e=(\ruleComposition{e_1}{e_2})$ with $\emptyset \typeOf e_1: \List\app \RuleType$.
By IH, we have that 
\begin{itemize}
\item $e_1 = v$. Then we have $\ruleComposition{v}{e_2}$ and \ref{r-rule-comp} applies. 
\item $e_1 = \errorLNC$. Then we have $\ruleComposition{\errorLNC}{e_2}$ and by \textsc{ctx-err} we have \\$V ; \LangDef \app ; \app \ruleComposition{\errorLNC}{e_2} \step V ; \LangDef \app ; \app \errorLNC$, for all $V, \LangDef$ because of the evaluation context $\ruleComposition{E}{e}$.
\item for all $V, \LangDef$, $V ; \LangDef ; e_1 \step V' ; \LangDef' ; e_1'$, for some $V',\LangDef',e'$. Then $\ruleComposition{e_1}{e_2}$ takes a step by \textsc{ctx-succ} or \textsc{ctx-lang-err}. 
\end{itemize}
\end{case}
\begin{case}[\textsc{t-seq}] 
Since (H1) then we have $e= e_1; e_2$ with $\emptyset \typeOf e_1: \LanguageType$.
By IH, we have that 
\begin{itemize}
\item $e_1 = v$. By Canonical form $e_1 = \skipLNC$. Then we have $\skipLNC; e_2$ which by  takes a step. 
\item $e_1 = \errorLNC$. Then we have $\errorLNC; e_2$ and by \textsc{ctx-err} we take a step to an error. 
\item for all $V, \LangDef$, $V ; \LangDef ; e_1 \step V' ; \LangDef' ; e_1'$, for some $V',\LangDef',e'$. Then by \textsc{ctx-succ} or \textsc{ctx-lang-err}, we take a step. 
\end{itemize}
\end{case}
\begin{case}[\textsc{t-selector}] 
Since (H1) then we have that $\emptyset \typeOf e_1: \List \app T$.
By IH, we have that 
\begin{itemize}
\item $e_1 = v$. By Canonical form $e_1$ can have two forms:
\begin{itemize} 
\item $e_1 =  \nil$ Then we apply \textsc{r-selector-nil} takes a step. 
\item $e_1 = \cons{\app v_1}{v_2}$. Then we have two cases: ether $\matchLNC{v_1,p}$ succeeds, then we apply \textsc{r-selector-cons-ok} and take a step, or $\matchLNC{v_1,p}$ fails, then we apply \textsc{r-selector-cons-fail} and take a step. 
\end{itemize}
\item $e_1 = \errorLNC$. Then by \textsc{ctx-err} we take a step to an error. 
\item for all $V, \LangDef$, $V ; \LangDef ; e_1 \step V' ; \LangDef' ; e_1'$, for some $V',\LangDef',e'$. Then by \textsc{ctx-succ} or \textsc{ctx-lang-err}, we take a step. 
\end{itemize}
The case for selectors with keep are analogous. 
\end{case}
\begin{case}[\textsc{t-uniquefy}] 
Since (H1) then we have that $\emptyset \typeOf e_1: \RuleType$, $\emptyset \typeOf e_2: \MapType{\OpnameType}{(\List\app\StringType)}$.
By IH on $e_1$, we have that 
\begin{itemize}
\item $e_1 = v_1$. By Canonical form $e_1 = r$. By IH on $e_2$ we have the three cases: 
\begin{itemize} 
\item $e_2 =  v_2$. Then we there are two cases: Either $\mathit{uniquefy}$ succeeds and we apply \textsc{r-uniquefy-ok} to take a step, or fails and we apply \textsc{r-uniquefy-fail} to take a step.
\item $e_2 = \errorLNC$. Then by \textsc{ctx-err} we take a step to an error. 
\item for all $V, \LangDef$, $V ; \LangDef ; e_2 \step V' ; \LangDef' ; e_2'$, for some $V',\LangDef',e_2'$. Then by \textsc{ctx-succ} or \textsc{ctx-lang-err}, we take a step. 
\end{itemize}
\item $e_1 = \errorLNC$. Then by \textsc{ctx-err} we take a step to an error. 
\item for all $V, \LangDef$, $V ; \LangDef ; e_1 \step V' ; \LangDef' ; e_1'$, for some $V',\LangDef',e_1'$. Then by \textsc{ctx-succ} or \textsc{ctx-lang-err}, we take a step. 
\end{itemize}
The case of $\emptyset \typeOf e_2: \MapType{\PrednameType}{(\List\app\StringType)}$ is analogous.
\end{case}
\begin{case}[\textsc{t-tick}] 
Since (H1) then we have that $\emptyset \typeOf e_1: T $, $\emptyset \typeOf e_2:\List\app\TermType$.
By IH on $e_1$, we have that 
\begin{itemize}
\item $e_1 = v_1$. By IH on $e_2$:
    \begin{itemize} 
    \item $e_2 = v_2$. Then we have two cases depending on $T$: 
        \begin{itemize} 
        \item $T =  \TermType$. By Canonical forms, we have that $e_1$ can be of the following forms:
            \begin{itemize} 
            \item $(opname \app v_1')$. Then we apply \ref{r-tick-opname} and take a step. 
            \item $\langVar{X}$. Then we apply \ref{r-tick-var} and take a step. 
            \item $(\langVar{z}v_1')$. Then we apply \ref{r-tick-abs} and take a step. 
            \item $v_1'[v_1''/\langVar{z}]$. Then we apply \ref{r-tick-sub} and take a step. 
            \end{itemize}
        \item $T = \List\app\TermType$. By Canonical form $e_1$ can have two forms:
            \begin{itemize} 
            \item $e_1 =  \nil$. Then we apply \ref{r-tick-nil} takes a step. 
            \item $e_1 = \cons{v_1}{v_2}$. Then we apply \ref{r-tick-cons} takes a step. 
            \end{itemize}
        \end{itemize}
    \item $e_2 = \errorLNC$. Then  by \textsc{ctx-err} we take a step to an error. 
    \item for all $V, \LangDef$, $V ; \LangDef ; e_2 \step V' ; \LangDef' ; e_2'$, for some $V',\LangDef',e_2'$. Then by \textsc{ctx-succ} or \textsc{ctx-lang-err}, we take a step. 
    \end{itemize}
    \item $e_1 = \errorLNC$. Then by \ref{ctx-err} we take a step to an error. 
    \item for all $V, \LangDef$, $V ; \LangDef ; e_1 \step V' ; \LangDef' ; e_1'$, for some $V',\LangDef',e_1'$. Then by \textsc{ctx-succ} or \textsc{ctx-lang-err}, we take a step. 
\end{itemize}
\end{case}

All other cases follow similar lines as above. 
\qed
\end{proof}

\begin{theorem}[Progress Theorem for Configurations]
For all , if $\emptyset \typeOf V ; \LangDef ; e$ then either 
\begin{itemize}
\item $e = \skipLNC$, or 
\item $e = \errorLNC$, or 
\item $V ; \LangDef ; e \step V' ; \LangDef' ; e'$, for some $e'$.
\end{itemize}
\end{theorem}
\begin{proof}
Let us assume the proviso: $\emptyset \typeOf V ; \LangDef ; e$. Then we have $\emptyset \typeOf  e : \LanguageType$.
By Progress Theorem for Expressions, we have that 
\begin{itemize}
\item $e = v$. By Canonical forms, $e = \skipLNC$.
\item $e = \errorLNC$, which satisfies the theorem. 
\item $V ; \LangDef ; e \step V' ; \LangDef' ; e'$, for some $e'$, which satisfies the theorem. 
\end{itemize}
\qed
\end{proof}

\subsection{Subject Reduction Theorem}

\begin{lemma}[Substitution Lemma]
if $\Gamma, x:T \typeOf e: T'$ and $\emptyset \typeOf v: T$ then $\Gamma \typeOf e[v/x]: T'$.
\end{lemma}
\begin{proof}
The proof is by induction on the derivation of $\Gamma, x:T \typeOf e: T'$. 
As usual, the case for variables \ruletag{t-var} relies on a standard weakening lemma: $\Gamma \typeOf e: T'$ and $x$ is not in the free variables of $e$ then $\Gamma, x:T \typeOf e: T'$, which can be proved by induction on the derivation of $\Gamma \typeOf e: T'$. 
An aspect that differs from a standard proof is that our substitution does not replace all instances of variables $\selfLNC$, $\premisesLNC$, and $\conclLNC$ in certain context. Then extra care must be taken in the substitution lemma because the substituted expression may still have those as free variables. The type system covers for those cases because it augments the type environment with $\GammaForRule{r}$. 

\qed
\end{proof}

\begin{lemma}[Pattern-matching typing and reduction]\label{thm:pattern-matching}
if $\emptyset \typeOfPattern p: T \Rightarrow \Gamma'$ and $\matchLNC{v}{p} = \theta$ then for all $x:T'\in \Gamma'$, $[x/v'] \in \theta$ and $\emptyset \typeOf v':T'$. 
\end{lemma}
\begin{proof}
The proof is by induction on the derivation of $\emptyset \typeOfPattern p: T \Rightarrow \Gamma'$. Each case is straightforward. 
\qed
\end{proof}

\begin{lemma}[$\uniquefySEMJustName\subListFormula$ produces well-typed results or fails]\label{uniquefyCorrectness}
\item $\uniquefySEM\subListFormula{lf}{m}{str}{mr} = res$ 
\begin{itemize}
\item $res =(lf',mr')$ such that\\ $\emptyset \typeOf lf' : \List\app\FormulaType$, and $\emptyset \typeOf mr' : \MapType{\VarType}{(\List\app\TermType)}$. 
\item $res = \failureSEM$. 
\end{itemize}
\end{lemma}
\begin{proof}
Straightforward induction on the definition of $\uniquefySEMJustName\subListFormula$. 
Most cases rely on the analogous lemmas for formulae, terms, list of terms and list of formulae: 
\begin{itemize}
\item $\uniquefySEM\subFormula{f}{m}{str}{mr} = res$ 
\begin{itemize}
\item $res =(f',mr')$ such that \\$\emptyset \typeOf f' : \FormulaType$, and $\emptyset \typeOf mr' : \MapType{\TermType}{(\List\app\TermType)}$. 
\item $res = \failureSEM$. 
\end{itemize}
\item $\uniquefySEM\subTerm{t}{m}{str}{mr} = res$ 
\begin{itemize}
\item $res =(t',mr')$ such that \\$\emptyset \typeOf t' : \TermType$, and $\emptyset \typeOf mr' : \MapType{\TermType}{(\List\app\TermType)}$. 
\item $res = \failureSEM$. 
\end{itemize}
\item $\uniquefySEM\subListTerm{lt}{m}{str}{mr} = res$ 
\begin{itemize}
\item $res =(lt',mr')$ such that \\$\emptyset \typeOf lt' : \List\app\TermType$, and $\emptyset \typeOf mr' : \MapType{\TermType}{(\List\app\TermType)}$. 
\item $res = \failureSEM$. 
\end{itemize}
\end{itemize}
Each can be proved with a straightforward induction on the definition of $\uniquefySEMJustName_{\mathcal{X}}$ where $\mathcal{X} \in \{\textsf{f,t,lt}\}$. 
\qed
\end{proof}

\begin{lemma}[Compositionality of $\typeOf$]\label{thm:compositionality}
if $\emptyset \typeOf E[e] : T$ then there exits $T'$ such that $\emptyset \typeOf e : T'$ and for all $e'$ if $\emptyset \typeOf e' : T'$ then $\emptyset \typeOf E[e'] : T$.
\end{lemma}
\begin{proof}
Proof is by induction on the structure of $E$. Each case is straightforward. 
\end{proof}

\begin{theorem}[Subject Reduction ($\step\atomic$)]
$V \cap \varsSEM{\LangDef} = \emptyset $, $\emptyset \typeOf  e : T$, and $V ; \LangDef \app ; \app e\step\atomic V' ; \LangDef' \app ; \app e'$ then $V' \cap \varsSEM{\LangDef'} = \emptyset$, and $\emptyset \typeOf  e' : T$. 
\end{theorem}

\begin{proof}
Let us assume the proviso of the theorem, that is, (H1) $V \cap \varsSEM{\LangDef} = \emptyset $, (H2) $\Gamma \typeOf  e : T$, and (H3) $V ; \LangDef \app ; \app e\step V' ; \LangDef' \app ; \app e'$. 
Case analysis on (H3). 
\begin{case}[\textsc{r-seq-ok}]
$V ; \LangDef \app ; \app (\skipLNC ; e) \step  V ; \LangDef \app ; \app e$. 
We need to prove 
$\varsSEM{\LangDef} \subseteq V$, which we already have by (H1).  We need to prove $\typeOf  \LangDef$, which we have by (H2). 
We have to prove that $\Gamma \typeOf  e: T$ where $\Gamma \typeOf  (\skipLNC ; e): T$. By \textsc{t-seq} we have that $\Gamma \typeOf  (\skipLNC ; e): \LanguageType$ (i.e. $T = \LanguageType$), $\Gamma \typeOf  e : \LanguageType$. 
\end{case}
%
%
\begin{case}[\textsc{r-newar}]
$V ; (G,R) \app ; \app \newVar\app str \step\atomic  V \cup \{\langVar{X'}\} ; \LangDef \app ; \app \langVar{X'}$. 
We need to prove 
$\varsSEM{\LangDef} \subseteq V \cup \{\langVar{X'}\}$, which we have because by (H1), we have that $\varsSEM{\LangDef} \subseteq V$, and we additionally we have that $\langVar{X'} \not\in \varsSEM{\LangDef}$. 
We have to prove that $\Gamma \typeOf  \langVar{X'}: \TermType$ because $\Gamma \typeOf  \newVar : \TermType$. This holds thanks to \textsc{t-metaVar}. 
\end{case}
\begin{case}[\textsc{r-rule-comp}]
$V ; \LangDef \app ; \app \ruleComposition{v}{e}  \step\atomic  V ; \LangDef \app ; \app e\subsForRule{v}$. 
We need to prove 
$V \cap \varsSEM{\LangDef} = \emptyset $, which we have because by (H1). 
We have to prove that (\#) $\emptyset \typeOf  e\subsForRule{v}:  \RuleType$ when (*) $\emptyset \typeOf  \ruleComposition{v}{e} : \RuleType$. From (*) we infer (HRULE) $\emptyset \typeOf v : \RuleType$ and $\GammaForRule{v} \typeOf \app e : \RuleType$, that is (HE) $\selfLNC : \RuleType, \premisesLNC : \List \app \FormulaType,  \conclLNC: \FormulaType \typeOf \app e : \RuleType$. 
By Canonical Form Lemma, from (HRULE) we infer that $v = \inference{v_1}{v_2}$ and since it is typeable (HRULE), by \textsc{t-rule} have $\emptyset\typeOf v_1: \List \app \FormulaType$ and $\emptyset\typeOf v_2 :\FormulaType$.
$e\subsForRule{v} = e[r/\selfLNC,v_1/ \premisesLNC  , v_2/\conclLNC  ] = e[v/\selfLNC][v_1/\premisesLNC][v_2/\conclLNC]$.

Given (HE), and given (HRULE), by Substitution Lemma we have ($HE_1$) $ \premisesLNC : \List \app \FormulaType,  \conclLNC: \FormulaType \typeOf \app e[v/\selfLNC] : \RuleType$.

Given ($HE_1$), and given $\Gamma\typeOf v_1: \List \app \FormulaType$ , by Substitution Lemma we have ($HE_2$) $  \conclLNC: \FormulaType \typeOf \app e[v/\selfLNC][v_1/\premisesLNC] : \RuleType$.

Given ($HE_2$), and given $\emptyset \typeOf v_2:  \FormulaType$ , by Substitution Lemma we have  $\emptyset \typeOf \app e[v/\selfLNC][v_1/\premisesLNC][v_2/\conclLNC] : \RuleType$.
\end{case}
\begin{case}[\textsc{r-selector-nil}]
$V ; \LangDef \app ; \app \select{\emptyListLNC}{p}{e} \step\atomic  V ; \LangDef \app ; \app \emptyListLNC$. 
We need to prove 
$V \cap \varsSEM{\LangDef} = \emptyset $, which we have because by (H1). 
We have to prove that (\#) $\Gamma \typeOf  \emptyListLNC:  \List\app T$ when (*) $\Gamma \typeOf  \select{\emptyListLNC}{p}{e} : \List\app T$. Thanks to \textsc{t-emptyList} this holds. 
\end{case}
\begin{case}[\textsc{r-selector-cons-ok}]
$V ; \LangDef \app ; \app \select{(\consLNC{v_1}{v_2})}{p}{e} \step\atomic $\\
$ V ; \LangDef \app ; \app (\consLNCFilterName\app {e\theta\theta'}\app {(\select{v_2}{p}{e})})$.
We need to prove 
$V \cap \varsSEM{\LangDef} = \emptyset $, which we have because by (H1). 
We have to prove that (\#) $\emptyset \typeOf  (\consLNCFilterName\app {e\theta\theta'}\app {(\select{v_2}{p}{e})}) : \List\app T'$ when (*) $\Gamma \typeOf  \select{(\consLNC{v_1}{v_2})}{p}{e} : \List\app T'$. 
By \textsc{t-selector} we have that $\Gamma \typeOf (\consLNC{v_1}{v_2}) :  \List\app T$, and therefore by \textsc{t-cons}, we have that $\Gamma \typeOf v_1:T$. We do a case analysis on whether $T=\RuleType$ or not, to prove in both cases that $\Gamma \typeOf e\theta\theta' : \MaybeType{T'}$. 
\begin{itemize}
\item $T=\RuleType$: By Canonical Form, then we have that $v_1 = \inference{v_1'}{v_2'}$. Then $\theta' =  [v_1/\selfLNC,v_1'/ \premisesLNC  , v_2'/\conclLNC  ]$. From (*) we infer that $ \Gamma', \selfLNC : \RuleType, \premisesLNC : \List \app \FormulaType,  \conclLNC: \FormulaType \typeOf \app e_2 : \MaybeType{T'}$, where $\Gamma'$ comes from the pattern-matching. 
By applying the same reasoning as in \ref{r-rule-comp}, we can apply the Substitution lemma three times to have $\Gamma' \typeOf e\theta' : \MaybeType{T'}$. By Lemma \ref{thm:pattern-matching} (pattern-matching correctness) we have that for all $(x:T'')\in\Gamma'$ there is $[x/v'']\in\theta$ such that $\emptyset\typeOf v'' : T''$. Then, for all such $(x:T'')\in\Gamma'$ we can use the Substitution Lemma to substitute its $[x/v'']$, and end up with $\emptyset \typeOf e\theta\theta' : \MaybeType{T'}$.
\item $T\not=\RuleType$: Then $\theta' =  [v_1/\selfLNC]$ and by Substitution lemma we have $\Gamma' \typeOf e\theta' : \MaybeType{T'}$. By pattern-matching correctness, the same reasoning as in the previous case leads us to $\emptyset \typeOf e\theta\theta' : \MaybeType{T'}$.
\end{itemize}
As now we know that (*) $\Gamma \typeOf e\theta\theta' : \MaybeType{T'}$ in all cases. 
If we expand $(\consLNCFilterName\app {e\theta\theta'}\app {(\select{v_2}{p}{e})})$ we have \\
$\consLNCFilter{e\theta\theta'}{(\select{v_2}{p}{e})}$. 
Here $\key{isNothing}$ and $\key{get}$ are applied to $e\theta\theta'$ of type $\MaybeType{T'}$, therefore are well-typed. Also, both branches of the if return an expression of type $\List \app T'$.
\end{case}
\begin{case}[\textsc{r-uniquefy-ok}]
$V ; \LangDef \app ; \app \uniquefy{\mathit{lf}}{v_1}{str}{x}{y}{e} \step\atomic  V ; \LangDef \app ; \app e[\mathit{lf}'/x,v_2/y] $. 
We need to prove 
$V \cap \varsSEM{\LangDef} = \emptyset $, which we have because by (H1). 
We have to prove that (\#) $\Gamma \typeOf e[\mathit{lf}'/x,v_2/y]: T$ when (*) $\Gamma \typeOf \uniquefy{r}{v}{str}{x}{y}{e} : T$.
By \textsc{r-uniquefy-ok} we have $(r', m) = \uniquefySEM\subRule{r}{v}{str}{\emptyMapLNC}$, and by Lemma \ref{uniquefyCorrectness} we have that \\
$\emptyset \typeOf \mathit{lf'} : \List \app\FormulaType$, and $\emptyset \typeOf v_2 : \MapType{\TermType}{(\List\app\TermType)}$. 
By \textsc{t-uniquefy} we have that $\Gamma, x: \List\app\FormulaType,  y: \MapType{\TermType}{(\List\app \TermType)} \typeOf \app e_3 : T $. 
By Substitution Lemma, we then have $\Gamma \typeOf e[\mathit{lf}'/x,v_2/y] : T$. 
\end{case}


All other cases are analogous. 
\end{proof}

\begin{theorem}[Subject Reduction ($\step$)]
For all $V$, $V'$, $\LangDef$, $\LangDef'$, $e$, $e'$, if $\emptyset \typeOf V ; \LangDef ; e$ and $V ; \LangDef ; e \step V' ; \LangDef' ; e'$ then $\emptyset \typeOf V' ; \LangDef' ; e'$.
\end{theorem}

\begin{proof}
Let us assume the proviso of the theorem and have (H1) $\emptyset \typeOf V ; \LangDef ; e$ and $V ; \LangDef ; e \step V' ; \LangDef' ; e'$.
The proof is by case analysis on the derivation of $V ; \LangDef ; e \step V' ; \LangDef' ; e'$

\begin{case}[\textsc{ctx-succ}]
$V ; \LangDef ; E[e] \step V' ; \LangDef' ; E[e']$ when (H2) $V ; \LangDef ; e \step\atomic V' ; \LangDef' ; e'$. 
From (H1) we know that (H6) $\emptyset \typeOf E[e] : T$, for some $T$. Then we can apply Lemma \ref{thm:compositionality} to have that (H3) $\emptyset \typeOf e : T'$, for some $T'$.
With (H2) and (H3) we can apply Subject Reduction for $\step\atomic$ and obtain that (H4) $\emptyset \typeOf e : T'$ and (H5) $V'\cap \varsSEM{\LangDef'} = \emptyset$. 
By Lemma \ref{thm:compositionality}, since we have (H6) and (H4) we can derive  $\emptyset \typeOf E[e'] : T$, and since we have (H5), we can derive $\emptyset \typeOf V ; \LangDef ; E[e']$.
\end{case}
\begin{case}[\textsc{ctx-lang-err}]
$V ; \LangDef \app ; \app E[e]\step V ; \LangDef \app ; \app \errorLNC$. (H1) implies $V \cap \varsSEM{\LangDef} = \emptyset$ and $\typeOf  \LangDef$. We need to prove 
$V \cap \varsSEM{\LangDef} = \emptyset$, which we already have, and $\typeOf  \LangDef$, which we already have. We need to prove $\Gamma \typeOf  \errorLNC: \LanguageType$, which we can prove with \textsc{(t-error)}. 
\end{case}
\begin{case}[\textsc{ctx-err}]
Similar lines as \textsc{ctx-lang-err}. 
\end{case}
\end{proof}

\subsection{Type Soundness}

\begin{theorem}[Type Soundness]
For all $\Gamma$, $V$, $\LangDef$, $e$, if $\typeOf V ; \LangDef ; e$ then $V ; \LangDef ; e \step^{*} V' ; \LangDef' ; e'$ s.t. i) $e' = \skipLNC$, ii) $e' = \errorLNC$, or iii) $V' ; \LangDef' ; e' \step V'' ; \LangDef'' ; e''$, for some $e''$.
\end{theorem}

The proof is straightforward once we have the Subject Reduction ($\step$) theorem, and the Progress for Configuration theorem, and that typeability is preserved in multiple steps (provable by straightforward induction on the derivation of $\step^{*}$).

\section{Let-Binding and Match in $\calcName$}\label{let}

$\key{let} \app x = e_1 \app \key{in}\app e_2 \equiv \key{head}\app (\select{[e_1]}{x}{e_2}) $\\

\noindent The pattern-matching that we use is unary-branched and either succeeds or throws an error. 

$\key{match} \app e_1 \app \key{with}\app p \Rightarrow e_2 \equiv \\
\key{let} \app x = (\select{[e_1]}{p}{e_2}))  \app \key{in}\app \key{if}\app (\key{isEmpty} \app x)\app \key{then}\app \key{error} \app\key{else}\app \key{head}\app x$
\end{document}

The proof is by induction on the derivation of $\emptyset \typeOf e : T$. The proof follows standard lines. 

\begin{theorem}[Subject Reduction ($\step\atomic$)]
$V \cap \varsSEM{\LangDef} = \emptyset $, $\emptyset \typeOf  e : T$, and $V ; \LangDef \app ; \app e\step\atomic V' ; \LangDef' \app ; \app e'$ then $V' \cap \varsSEM{\LangDef'} = \emptyset$, and $\emptyset \typeOf  e' : T$. 
\end{theorem}

\section{Meta-theory}

\begin{lemma}[Substitution Lemma]
if $\Gamma, x:T \typeOf e: T'$ and $\emptyset \typeOf v: T$ then $\Gamma \typeOf e[v/x]: T'$.
\end{lemma}
\begin{proof}
The proof is by induction on the derivation of $\Gamma, x:T \typeOf e: T'$. 
As usual, the case for variables (\ref{t-var}) relies on a standard weakening lemma: $\Gamma \typeOf e: T'$ and $x\not\in FV(e)$ then $\Gamma, x:T \typeOf e: T'$, which can be proved by induction on the derivation of $\Gamma \typeOf e: T'$. 
\qed
\end{proof}

\begin{lemma}[Pattern-matching typing and reduction]\label{thm:pattern-matching}
if $\emptyset \typeOfPattern p: T \Rightarrow \Gamma'$ and $\matchLNC{v}{p} = \theta$ then for all $x:T'\in \Gamma'$, $[x/v'] \in \theta$ and $\emptyset \typeOf v':T'$. 
\end{lemma}
\begin{proof}
The proof is by induction on the derivation of $\emptyset \typeOfPattern p: T \Rightarrow \Gamma'$. Each case is straightforward. 
\qed
\end{proof}

\begin{lemma}[$\uniquefySEMJustName\subRule$ produces well-typed results or fails]\label{uniquefyCorrectness}
$\uniquefySEM\subRule{r}{m}{str}{mr} = res$ with 
\begin{itemize}
\item $res =(r',mr')$ such that $\emptyset \typeOf r' : \RuleType$, and $\emptyset \typeOf mr' : \MapType{\VarType}{(\List\app\VarType)}$. 
\item $res = \failureSEM$. 
\end{itemize}
\end{lemma}
\begin{proof}
Straightforward induction on the definition of $\uniquefySEMJustName$ of Figure \ref{fig:uniquefy}. 
Most cases rely on the analogous lemmas for formulae, terms, list of terms and list of formulae: 
\begin{itemize}
\item $\uniquefySEM\subFormula{f}{m}{str}{mr} = res$ 
\begin{itemize}
\item $res =(f',mr')$ such that $\emptyset \typeOf f' : \FormulaType$, and $\emptyset \typeOf mr' : \MapType{\VarType}{(\List\app\VarType)}$. 
\item $res = \failureSEM$. 
\end{itemize}
\item $\uniquefySEM\subListFormula{lf}{m}{str}{mr} = res$ 
\begin{itemize}
\item $res =(lf',mr')$ such that $\emptyset \typeOf lf' : \List\app\FormulaType$, and $\emptyset \typeOf mr' : \MapType{\VarType}{(\List\app\VarType)}$. 
\item $res = \failureSEM$. 
\end{itemize}
\item $\uniquefySEM\subTerm{t}{m}{str}{mr} = res$ 
\begin{itemize}
\item $res =(t',mr')$ such that $\emptyset \typeOf t' : \TermType$, and $\emptyset \typeOf mr' : \MapType{\VarType}{(\List\app\VarType)}$. 
\item $res = \failureSEM$. 
\end{itemize}
\item $\uniquefySEM\subListTerm{lt}{m}{str}{mr} = res$ 
\begin{itemize}
\item $res =(lt',mr')$ such that $\emptyset \typeOf lt' : \List\app\TermType$, and $\emptyset \typeOf mr' : \MapType{\VarType}{(\List\app\VarType)}$. 
\item $res = \failureSEM$. 
\end{itemize}
\end{itemize}
Each can be proved with a straightforward induction on the definition of $\uniquefySEMJustName_{\mathcal{X}}$ where $\mathcal{X} \in \{\textsf{f,lf,t,lt}\}$. 
\qed
\end{proof}

\begin{lemma}[Compositionality of $\typeOf$]\label{thm:compositionality}
if $\emptyset \typeOf E[e] : T$ then there exits $T'$ such that $\emptyset \typeOf e : T'$ and for all $e'$ if $\emptyset \typeOf e' : T'$ then $\emptyset \typeOf E[e'] : T$.
\end{lemma}
\begin{proof}
Proof is by induction on the structure of $E$. Each case is straightforward. 
\end{proof}

\begin{theorem}[Subject Reduction ($\step\atomic$)]
$V \cap \varsSEM{\LangDef} = \emptyset $, $\emptyset \typeOf  e : T$, and $V ; \LangDef \app ; \app e\step\atomic V' ; \LangDef' \app ; \app e'$ then $V' \cap \varsSEM{\LangDef'} = \emptyset$, and $\emptyset \typeOf  e' : T$. 
\end{theorem}

\begin{theorem}[Subject Reduction ($\step$)]
For all $V$, $V'$, $\LangDef$, $\LangDef'$, $e$, $e'$, if $\emptyset \typeOf V ; \LangDef ; e$ and $V ; \LangDef ; e \step V' ; \LangDef' ; e'$ then $\emptyset \typeOf V' ; \LangDef' ; e'$.
\end{theorem}

\begin{definition}[Substitutions for Category Names]
\begin{itemize}
\item $cname : \List\app\TermType \in \GammaForCnames{G}$ iff $cname \app \langVar{X} ::= lt \in G$.
\item $\subsForCnames{G}$ is the composition of all substitutions $[lt/cname]$ such that $cname \app \langVar{X} ::= lt \in G$.
\end{itemize}
\end{definition}

\begin{theorem}[Type Soundness]
For all $\Gamma$, $V$, $\LangDef = (G,R)$, $e$, if $\GammaForCnames{G} \typeOf V ; \LangDef ; e$ then $V ; \LangDef ; e\subsForCnames{G} \step^{*} V' ; \LangDef' ; e'$ such that  
\begin{itemize}
\item $e' = \skipLNC$, or 
\item $e' = \errorLNC$, or 
\item $V' ; \LangDef' ; e' \step V'' ; \LangDef'' ; e''$, for some $e''$.
\end{itemize}
\end{theorem}

	                     	V ; \LangDef \app ; \app \select{\emptyListLNC(\key{keep})}{p}{e} \step\atomic  V ; \LangDef \app ; \app \emptyListLNC 	\label{beta}\tagsc{r-selector-nil-keep}\\
	\inference{\matchLNC{v_1}{p} = \theta \\ 
	\theta' = \mbox{$\begin{cases}
	                              \subsForRule{r}
	                              & \mbox{if } v_1=r\\ 
	                          \{\selfLNC \mapsto v_1\} & \mbox{otherwise}
	                        \end{cases}
	                     $}}
	                     {V ; \LangDef \app ; \app \select{(\consLNC{v_1}{v_2})(\key{keep})}{p}{e} \step\atomic  V ; \LangDef \app ; \app \filterNothingLNC{(\consLNC{e\theta\theta'}{(\select{v_2}{p}{e}}))}} 	\label{beta}\tagsc{r-selector-cons-keep-ok}\\
	\inference{\matchLNC{v_1}{p} \not= \theta}  
	                     {V ; \LangDef \app ; \app \select{(\consLNC{v_1}{v_2})(\key{keep})}{p}{e} \step\atomic  V ; \LangDef \app ; \app \filterNothingLNC{((\consLNC{v_1}{(\select{v_2}{p}{e})}))}	}\label{beta}\tagsc{r-selector-cons-keep-fail}\\

\appendix 

\section{Proof of Subject Reduction}

\begin{theorem}[Subject Reduction ($\step\atomic$)]
$V \cap \varsSEM{\LangDef} = \emptyset $, $\emptyset \typeOf  e : T$, and $V ; \LangDef \app ; \app e\step\atomic V' ; \LangDef' \app ; \app e'$ then $V' \cap \varsSEM{\LangDef'} = \emptyset$, and $\emptyset \typeOf  e' : T$. 
\end{theorem}

\begin{proof}
Let us assume the proviso of the theorem, that is, (H1) $V \cap \varsSEM{\LangDef} = \emptyset $, (H2) $\Gamma \typeOf  e : T$, and (H3) $V ; \LangDef \app ; \app e\step V' ; \LangDef' \app ; \app e'$. 
Case analysis on (H3). 
\begin{case}[\ref{r-seq-ok}]
$V ; \LangDef \app ; \app (\skipLNC ; e) \step  V ; \LangDef \app ; \app e$. 
We need to prove 
$\varsSEM{\LangDef} \subseteq V$, which we already have by (H1).  We need to prove $\typeOf  \LangDef$, which we have by (H2). 
We have to prove that $\Gamma \typeOf  e: T$ where $\Gamma \typeOf  (\skipLNC ; e): T$. By \ref{t-seq} we have that $\Gamma \typeOf  (\skipLNC ; e): \LanguageType$ (i.e. $T = \LanguageType$), $\Gamma \typeOf  e : \LanguageType$. 
\end{case}
\begin{case}[\ref{r-rules-ok}]
$V ; (G,R) \app ; \app R' \step  V ; (G,R') \app ; \app \skipLNC$, with $\varsSEM{R'} \subseteq V$. 
We need to prove 
$\varsSEM{(G,R')} \subseteq V$, which we have because $\varsSEM{R'} \subseteq V$. 
We have to prove that $\Gamma \typeOf  \skipLNC: \List\app \RuleType$ or $\Gamma \typeOf  \skipLNC:  \LanguageType$ when $\Gamma \typeOf  R' : \List\app \RuleType$ or $\Gamma \typeOf  R' : \LanguageType$, respectively. This holds thanks to \ref{t-skip}. 
\end{case}
\begin{case}[\ref{r-rules-fail}]
$V ; (G,R) \app ; \app R'  \step\atomic  V ; (G,R) \app ; \app \errorLNC$. 
We need to prove 
$\varsSEM{(G,R)} \subseteq V$, which we have by (H1). 
We have to prove that $\Gamma \typeOf  \errorLNC: \List\app \RuleType$ or $\Gamma \typeOf  \errorLNC:  \LanguageType$ when $\Gamma \typeOf  R' : \List\app \RuleType$ or $\Gamma \typeOf  R' : \LanguageType$, respectively. This holds thanks to \ref{t-error}. 
\end{case}
\begin{case}[\ref{r-newar}]
$V ; (G,R) \app ; \app \newVar\app str \step\atomic  V \cup \{\langVar{X'}\} ; \LangDef \app ; \app \langVar{X'}$. 
We need to prove 
$\varsSEM{\LangDef} \subseteq V \cup \{\langVar{X'}\}$, which we have because by (H1), we have that $\varsSEM{\LangDef} \subseteq V$, and we additionally we have that $\langVar{X'} \not\in \varsSEM{\LangDef}$. 
We have to prove that $\Gamma \typeOf  \langVar{X'}: \VarType$ because $\Gamma \typeOf  R' : \List\app \RuleType$ or $\Gamma \typeOf  \newVar : \VarType$. This holds thanks to \ref{t-lang-var}. 
\end{case}
\begin{case}[\ref{r-unbound-abs}]
$V ; \LangDef \app ; \app \unboundLNC{(\langVar{x})t} \step\atomic  V ; \LangDef \app ; \app t$. 
We need to prove 
$\varsSEM{\LangDef} \subseteq V $, which we have because by (H1). 
We have to prove that $\Gamma \typeOf  t: \TermType$ because $\Gamma \typeOf  \unboundLNC{(\langVar{x})t}: \TermType$. By \ref{t-unbound} we have $\Gamma \typeOf  (\langVar{x})t : \TermType$. Now, by \ref{t-var-abs} we have $\Gamma \typeOf t : \TermType$. 
\end{case}
\begin{case}[\ref{r-unbound-noAbs}]
$V ; \LangDef \app ; \app \unboundLNC{t} \step\atomic  V ; \LangDef \app ; \app t$. 
We need to prove 
$\varsSEM{\LangDef} \subseteq V $, which we have because by (H1). 
We have to prove that $\Gamma \typeOf  t: \TermType$ because $\Gamma \typeOf  \unboundLNC{t}: \TermType$. By \ref{t-unbound} we have $\Gamma \typeOf  t : \TermType$. 
\end{case}
\begin{case}[\ref{r-fold-0}]
$V ; \LangDef \app ; \app \foldLNC{predname}{\emptyListLNC} \step\atomic  V ; \LangDef \app ; \app \emptyListLNC$. 
We need to prove 
$\varsSEM{\LangDef} \subseteq V $, which we have because by (H1). 
We have to prove that $\Gamma \typeOf  \emptyListLNC: \List \app \FormulaType$ because $\Gamma \typeOf  \foldLNC{predname}{\emptyListLNC} : \List \app \FormulaType$. This holds thanks to \ref{t-nil}.
\end{case}
\begin{case}[\ref{r-fold-1}]
$V ; \LangDef \app ; \app \foldLNC{predname}{(\consLNC{v}{\emptyListLNC})} \step\atomic  V ; \LangDef \app ; \app \emptyListLNC$. 
We need to prove 
$\varsSEM{\LangDef} \subseteq V $, which we have because by (H1). 
We have to prove that $\Gamma \typeOf  \emptyListLNC: \List \app \FormulaType$ because $\Gamma \typeOf  \foldLNC{predname}{(\consLNC{v}{\emptyListLNC})} : \List \app \FormulaType$. This holds thanks to \ref{t-nil}.
\end{case}
\begin{case}[\ref{r-fold-2}]
$V ; \LangDef \app ; \app \foldLNC{predname}{(\consLNC{v_1}{(\consLNC{v_2}{v_3})})} \step\atomic V ; \LangDef \app ; \app (\consLNC{(predname \app [v_1, v_2])}(\foldLNC{predname}{(\consLNC{v_2}{v_3})}))$. (Where $[v_1, v_2] = (\consLNC{v_1}{(\consLNC{v_2}{\emptyListLNC})})$).
We need to prove 
$\varsSEM{\LangDef} \subseteq V $, which we have because by (H1). 
We have to prove that (\#) $\Gamma \typeOf (\consLNC{(predname \app v_1\app v_2)}(\foldLNC{predname}{(\consLNC{v_2}{v_3})})): \List \app \FormulaType$ because (*) $\Gamma \typeOf  \foldLNC{predname}{(\consLNC{v_1}{(\consLNC{v_2}{v_3})})} : \List \app \FormulaType$. By \ref{t-fold}, from (*) we infer that $\Gamma \typeOf (\consLNC{v_1}{(\consLNC{v_2}{v_3})}) : \List \app \TermType$. In turn, by \ref{t-cons}, applied twice, we infer that $\Gamma \typeOf v_1 : \TermType$, $\Gamma \typeOf v_2 : \TermType$ and $\Gamma \typeOf v_3 : \List \app \TermType$. Now, by \ref{t-fold}, we have that  $\Gamma \typeOf (\foldLNC{predname}{(\consLNC{v_2}{v_3})})): \List \app \FormulaType)$ because $\Gamma \typeOf (\consLNC{v_2}{v_3}) : \List\app \TermType$. By \ref{t-formula} we have that  $\Gamma \typeOf (predname \app [v_1, v_2]) : \FormulaType$, hence (\#) holds. 
\end{case}
\begin{case}[\ref{r-getOverlap}]
$V ; \LangDef \app ; \app\getOverlap{v_1}{v_2} \step\atomic  V ; \LangDef \app ; \app \varsSEM{v_1} \cap \varsSEM{v_2}$. 
We need to prove 
$\varsSEM{\LangDef} \subseteq V $, which we have because by (H1). 
We have to prove that (\#) $\Gamma \typeOf  \varsSEM{v_1} \cap \varsSEM{v_2}: \List \app \VarType$ because $\Gamma \typeOf  \getOverlap{v_1}{v_2} : \List \app \VarType$. We can apply $\varsSEM{\cdot}$ because by \ref{t-getOverlap} we have that $\Gamma \typeOf v_1 : \TermType$ and $\Gamma \typeOf v_2 : \TermType$. Our meta-function $\cap$ returns a list of variables, hence (\#) holds. 
\end{case}
\begin{case}[\ref{r-filter-nil}]
$V ; \LangDef \app ; \app \filterNothingLNC{\emptyListLNC} \step\atomic  V ; \LangDef \app ; \app \emptyListLNC$. 
We need to prove 
$\varsSEM{\LangDef} \subseteq V $, which we have because by (H1). 
We have to prove that $\Gamma \typeOf  \emptyListLNC: \List \app T$ when $\Gamma \typeOf  \filterNothingLNC{\emptyListLNC}: \List \app  T$. This holds thanks to \ref{t-nil}.
\end{case}
\begin{case}[\ref{r-filter-nil}]
$V ; \LangDef \app ; \app \filterNothingLNC{(\consLNC{\nothingLNC}{v_2})} \step\atomic  V ; \LangDef \app ; \app \filterNothingLNC{v_2}$. 
We need to prove 
$\varsSEM{\LangDef} \subseteq V $, which we have because by (H1). 
We have to prove that (\#) $\Gamma \typeOf   \filterNothingLNC{v_2}: \List \app T$ when (*) $\filterNothingLNC{(\consLNC{\nothingLNC}{v_2})}: \List \app  T$. By \ref{t-filterNothing}, from (*) we have that $\Gamma \typeOf (\consLNC{\nothingLNC}{v_2}) : \List \app \MaybeType{T}$, hence $\Gamma \typeOf v_2 :  \List \app \MaybeType{T}$. We then can apply \ref{t-filterNothing} to have that (\#) holds.
\end{case}
\begin{case}[\ref{r-filter-cons-just}]
$V ; \LangDef \app ; \app \filterNothingLNC{(\consLNC{(\justLNC{v_1})}{v_2})} \step\atomic  V ; \LangDef \app ; \app (\consLNC{v_1}{(\filterNothingLNC{v_2})})$. 
We need to prove 
$\varsSEM{\LangDef} \subseteq V $, which we have because by (H1). 
We have to prove that (\#) $\Gamma \typeOf   (\consLNC{v_1}{(\filterNothingLNC{v_2})}): \List \app T$ when (*) $\filterNothingLNC{(\consLNC{(\justLNC{v_1})}{v_2})}: \List \app  T$. By \ref{t-filterNothing}, from (*) we have that $\Gamma \typeOf (\consLNC{(\justLNC{v_1})}{v_2}) : \List \app \MaybeType{T}$, hence $\Gamma \typeOf v_2 :  \List \app \MaybeType{T}$, and $\Gamma \typeOf v_1:T$. Then by \ref{t-filterNothing} we have that $\Gamma \typeOf \filterNothingLNC{v_2} :  \List \app T$. Then, by \ref{t-cons}, (\#) holds. 
\end{case}
\begin{case}[\ref{r-rule-comp}]
$V ; \LangDef \app ; \app \ruleComposition{v}{e}  \step\atomic  V ; \LangDef \app ; \app e\subsForRule{v}$. 
We need to prove 
$V \cap \varsSEM{\LangDef} = \emptyset $, which we have because by (H1). 
We have to prove that (\#) $\emptyset \typeOf  e\subsForRule{v}:  \RuleType$ when (*) $\emptyset \typeOf  \ruleComposition{v}{e} : \RuleType$. From (*) we infer (HRULE) $\emptyset \typeOf v : \RuleType$ and $\GammaForRule{v} \typeOf \app e : \RuleType$, that is (HE) $\selfLNC : \RuleType, \premisesLNC : \List \app \FormulaType,  \conclLNC: \FormulaType \typeOf \app e : \RuleType$. 
By Canonical Form Lemma, from (HRULE) we infer that $v = \inference{v_1}{v_2}$ and since it is typeable (HRULE), by \ref{t-rule} have $\emptyset\typeOf v_1: \List \app \FormulaType$ and $\emptyset\typeOf v_2 :\FormulaType$.
$e\subsForRule{v} = e[r/\selfLNC,v_1/ \premisesLNC  , v_2/\conclLNC  ] = e[v/\selfLNC][v_1/\premisesLNC][v_2/\conclLNC]$.

Given (HE), and given (HRULE), by Substitution Lemma we have ($HE_1$) $ \premisesLNC : \List \app \FormulaType,  \conclLNC: \FormulaType \typeOf \app e[v/\selfLNC] : \RuleType$.

Given ($HE_1$), and given $\Gamma\typeOf v_1: \List \app \FormulaType$ , by Substitution Lemma we have ($HE_2$) $  \conclLNC: \FormulaType \typeOf \app e[v/\selfLNC][v_1/\premisesLNC] : \RuleType$.

Given ($HE_2$), and given $\emptyset \typeOf v_2:  \FormulaType$ , by Substitution Lemma we have  $\emptyset \typeOf \app e[v/\selfLNC][v_1/\premisesLNC][v_2/\conclLNC] : \RuleType$.
\end{case}
\begin{case}[\ref{r-selector-nil}]
$V ; \LangDef \app ; \app \select{\emptyListLNC}{p}{e} \step\atomic  V ; \LangDef \app ; \app \emptyListLNC$. 
We need to prove 
$V \cap \varsSEM{\LangDef} = \emptyset $, which we have because by (H1). 
We have to prove that (\#) $\Gamma \typeOf  \emptyListLNC:  \List\app T$ when (*) $\Gamma \typeOf  \select{\emptyListLNC}{p}{e} : \List\app T$. Thanks to \ref{t-emptyList} this holds. 
\end{case}
\begin{case}[\ref{r-selector-cons-ok}]
$V ; \LangDef \app ; \app \select{(\consLNC{v_1}{v_2})}{p}{e} \step\atomic  V ; \LangDef \app ; \app \filterNothingLNC{(\consLNC{e\theta\theta'}{(\select{v_2}{p}{e}}))}$. 
We need to prove 
$V \cap \varsSEM{\LangDef} = \emptyset $, which we have because by (H1). 
We have to prove that (\#) $\emptyset \typeOf  \filterNothingLNC{(\consLNC{e\theta\theta'}{(\select{v_2}{p}{e}}))} : \List\app T'$ when (*) $\Gamma \typeOf  \select{(\consLNC{v_1}{v_2})}{p}{e} : \List\app T'$
By \ref{t-selector} we have that $\Gamma \typeOf (\consLNC{v_1}{v_2}) :  \List\app T$, and therefore by \ref{t-cons}, we have that $\Gamma \typeOf v_1:T$. We do a case analysis on whether $T=\RuleType$ or not, to prove in both cases that $\Gamma \typeOf e\theta\theta' : \MaybeType{T'}$. 
\begin{itemize}
\item $T=\RuleType$: By Canonical Form, then we have that $v_1 = \inference{v_1'}{v_2'}$. Then $\theta' =  [v_1/\selfLNC,v_1'/ \premisesLNC  , v_2'/\conclLNC  ]$. From (*) we infer that $ \Gamma', \selfLNC : \RuleType, \premisesLNC : \List \app \FormulaType,  \conclLNC: \FormulaType \typeOf \app e_2 : \MaybeType{T'}$, where $\Gamma'$ comes from the pattern-matching. 
By applying the same reasoning as in \ref{r-rule-comp}, we can apply the Substitution lemma three times to have $\Gamma' \typeOf e\theta' : \MaybeType{T'}$. By Lemma \ref{thm:pattern-matching} (pattern-matching correctness) we have that for all $(x:T'')\in\Gamma'$ there is $[x/v'']\in\theta$ such that $\emptyset\typeOf v'' : T''$. Then, for all such $(x:T'')\in\Gamma'$ we can use the Substitution Lemma to substitute its $[x/v'']$, and end up with $\emptyset \typeOf e\theta\theta' : \MaybeType{T'}$.
\item $T\not=\RuleType$: Then $\theta' =  [v_1/\selfLNC]$ and by Substitution lemma we have $\Gamma' \typeOf e\theta' : \MaybeType{T'}$. By pattern-matching correctness, the same reasoning as in the previous case leads us to $\emptyset \typeOf e\theta\theta' : \MaybeType{T'}$.
\end{itemize}
As now we know that (*) $\Gamma \typeOf e\theta\theta' : \MaybeType{T'}$ in all cases. We have that $\Gamma \typeOf (\consLNC{e\theta\theta'}{(\select{v_2}{p}{e}})) : \List\app \MaybeType{T'}$ because of (*) and because $\Gamma \typeOf (\select{v_2}{p}{e}) :  \List\app \MaybeType{T'}$. Therefore, by \ref{t-filterNothing}, we have (\#). 
\end{case}
\begin{case}[\ref{r-uniquefy-ok}]
$V ; \LangDef \app ; \app \uniquefy{r}{v}{str}{x}{e} \step\atomic  V ; \LangDef \app ; \app e\subsForRule{r'}[m/x]$. 
We need to prove 
$V \cap \varsSEM{\LangDef} = \emptyset $, which we have because by (H1). 
We have to prove that (\#) $\Gamma \typeOf e\subsForRule{r'}[m/x] : T$ when (*) $\Gamma \typeOf \uniquefy{r}{v}{str}{x}{e} : T$.
By \ref{r-uniquefy-ok} we have $(r', m) = \uniquefySEM\subRule{r}{v}{str}{\emptyMapLNC}$, and by Lemma \ref{uniquefyCorrectness} we have that $\emptyset \typeOf r' : \RuleType$, and $\emptyset \typeOf m : \MapType{\VarType}{(\List\app\VarType)}$. 
By \ref{t-uniquefy} we have that $\Gamma, \GammaForRule{r}, x: \MapType{\VarType}{(\List\app \VarType)} \typeOf \app e: T$. 
By Substitution Lemma, we then have $\Gamma, \GammaForRule{r} \typeOf e[m/x] : T$. By application of Substitution Lemma three times, in the same way as done in \ref{r-rule-comp}, we have (\#). 

\end{case}
\end{proof}

\begin{theorem}[Subject Reduction ($\step$)]
For all $V$, $V'$, $\LangDef$, $\LangDef'$, $e$, $e'$, if $\emptyset \typeOf V ; \LangDef ; e$ and $V ; \LangDef ; e \step V' ; \LangDef' ; e'$ then $\emptyset \typeOf V' ; \LangDef' ; e'$.
\end{theorem}
\begin{proof}
Let us assume the proviso of the theorem and have (H1) $\emptyset \typeOf V ; \LangDef ; e$ and $V ; \LangDef ; e \step V' ; \LangDef' ; e'$.
The proof is by case analysis on the derivation of $V ; \LangDef ; e \step V' ; \LangDef' ; e'$

\begin{case}[\textsc{ctx-succ}]
$V ; \LangDef ; E[e] \step V' ; \LangDef' ; E[e']$ when (H2) $V ; \LangDef ; e \step\atomic V' ; \LangDef' ; e'$. 
From (H1) we know that (H6) $\emptyset \typeOf E[e] : T$, for some $T$. Then we can apply Lemma \ref{thm:compositionality} to have that (H3) $\emptyset \typeOf e : T'$, for some $T'$.
With (H2) and (H3) we can apply Subject Reduction for $\step\atomic$ and obtain that (H4) $\emptyset \typeOf e : T'$ and (H5) $V'\cap \varsSEM{\LangDef'} = \emptyset$. 
By Lemma \ref{thm:compositionality}, since we have (H6) and (H4) we can derive  $\emptyset \typeOf E[e'] : T$, and since we have (H5), we can derive $\emptyset \typeOf V ; \LangDef ; E[e']$.
\end{case}
\begin{case}[\textsc{ctx-lang-err}]
$V ; \LangDef \app ; \app E[e]\step V ; \LangDef \app ; \app \errorLNC$. We need to prove 
$\varsSEM{\LangDef} \subseteq V$, which we already have by (H1).  We need to prove $\typeOf  \LangDef$, which we have by (H2), and $\Gamma \typeOf  \errorLNC: T$, which we can prove with \textsc{(t-error)}. 
\end{case}
\begin{case}[\textsc{ctx-err}]
Same line as \textsc{ctx-lang-err}. 
\end{case}
\end{proof}

\section{Progress Theorem}

\begin{theorem}[Progress Theorem Expressions]
For all , if $\emptyset \typeOf e : T$ then either 
\begin{itemize}
\item $e = v$, or 
\item $e = \errorLNC$, or 
\item for all $V, \LangDef$, $V ; \LangDef ; e \step V' ; \LangDef' ; e'$, for some $V',\LangDef',e'$.
\end{itemize}
\end{theorem}
\begin{proof}
We prove the theorem by induction on the derivation of $\emptyset \typeOf e: T$. 
Let us assume the proviso of the theorem, that is (H1) $\emptyset \typeOf e: T$. 
\begin{case}[\ref{t-opname}] 
Since (H1) then we have $(opname \app e)$ with $\emptyset \typeOf e: \List\app \TermType$.
By IH, we have that 
\begin{itemize}
\item $e = v$. Then we have $(opname \app v)$ which is a value. 
\item $e = \errorLNC$. Then we have $(opname \app \errorLNC)$ and by \ref{ctx-err} we have $V ; \LangDef \app ; \app (opname \app \errorLNC) \step V ; \LangDef \app ; \app \errorLNC$, for all $V, \LangDef$ because of the evaluation context $(opname \app E)$.
\item for all $V, \LangDef$, $V ; \LangDef ; e \step V' ; \LangDef' ; e'$, for some $V',\LangDef',e'$. Then $(opname \app e)$ takes a step by \ref{ctx-succ} or \ref{ctx-lang-err}. 
\end{itemize}
\end{case}
\begin{case}[\ref{t-rule-comp}] 
Since (H1) then we have $(\ruleComposition{e_1}{e_2})$ with $\emptyset \typeOf e_1: \List\app \RuleType$.
By IH, we have that 
\begin{itemize}
\item $e_1 = v$. Then we have $\ruleComposition{v}{e_2}$ and \ref{r-rule-comp} applies. 
\item $e_1 = \errorLNC$. Then we have $\ruleComposition{\errorLNC}{e_2}$ and by \ref{ctx-err} we have $V ; \LangDef \app ; \app \ruleComposition{\errorLNC}{e_2} \step V ; \LangDef \app ; \app \errorLNC$, for all $V, \LangDef$ because of the evaluation context $\ruleComposition{E}{e}$.
\item for all $V, \LangDef$, $V ; \LangDef ; e_1 \step V' ; \LangDef' ; e_1'$, for some $V',\LangDef',e'$. Then $\ruleComposition{e_1}{e_2}$ takes a step by \ref{ctx-succ} or \ref{ctx-lang-err}. 
\end{itemize}
\end{case}
\begin{case}[\ref{t-seq}] 
Since (H1) then we have $e_1; e_2$ with $\emptyset \typeOf e_1: \LanguageType$.
By IH, we have that 
\begin{itemize}
\item $e_1 = v$. By Canonical form $e_1 = \skipLNC$. Then we have $\skipLNC; e_2$ which by \ref{r-seq-collapse} takes a step. 
\item $e_1 = \errorLNC$. Then we have $\errorLNC; e_2$ and by \ref{ctx-err} we take a step to an error. 
\item for all $V, \LangDef$, $V ; \LangDef ; e_1 \step V' ; \LangDef' ; e_1'$, for some $V',\LangDef',e'$. Then by \ref{ctx-succ} or \ref{ctx-lang-err}, we take a step. 
\end{itemize}
\end{case}
\begin{case}[\ref{t-selector}] 
Since (H1) then we have that $\emptyset \typeOf e_1: \List \app T$.
By IH, we have that 
\begin{itemize}
\item $e_1 = v$. By Canonical form $e_1$ can have two forms:
\begin{itemize} 
\item $e_1 =  \nil$ Then we apply \ref{r-selector-nil} takes a step. 
\item $e_1 = \cons{v_1}{v_2}$. Then we have two cases: ether $\matchLNC{v_1,p}$ succeeds, then we apply \ref{r-selector-cons-ok} and take a step, or $\matchLNC{v_1,p}$ fails, then we apply \ref{r-selector-cons-fail} and take a step. 
\end{itemize}
\item $e_1 = \errorLNC$. Then by \ref{ctx-err} we take a step to an error. 
\item for all $V, \LangDef$, $V ; \LangDef ; e_1 \step V' ; \LangDef' ; e_1'$, for some $V',\LangDef',e'$. Then by \ref{ctx-succ} or \ref{ctx-lang-err}, we take a step. 
\end{itemize}
\end{case}
\begin{case}[\ref{t-uniquefy}] 
Since (H1) then we have that $\emptyset \typeOf e_1: \RuleType$, $\emptyset \typeOf e_2: \MapType{\NameType}{(\List\app\StringType)}$.
By IH on $e_1$, we have that 
\begin{itemize}
\item $e_1 = v_1$. By Canonical form $e_1 = r$. By IH on $e_2$ we have the three cases: 
\begin{itemize} 
\item $e_2 =  v_2$. Then we there are two cases: Either $(r', m) = \uniquefySEM\subRule{r}{v}{str}{\emptyMapLNC}$, and we apply \ref{r-uniquefy-ok} to take a step, or $(r', m) \not= \uniquefySEM\subRule{r}{v}{str}{\emptyMapLNC}$, and we apply \ref{r-uniquefy-fail} to take a step.
\item $e_2 = \errorLNC$. Then by \ref{ctx-err} we take a step to an error. 
\item for all $V, \LangDef$, $V ; \LangDef ; e_2 \step V' ; \LangDef' ; e_2'$, for some $V',\LangDef',e_2'$. Then by \ref{ctx-succ} or \ref{ctx-lang-err}, we take a step. 
\end{itemize}
\item $e_1 = \errorLNC$. Then by \ref{ctx-err} we take a step to an error. 
\item for all $V, \LangDef$, $V ; \LangDef ; e_1 \step V' ; \LangDef' ; e_1'$, for some $V',\LangDef',e_1'$. Then by \ref{ctx-succ} or \ref{ctx-lang-err}, we take a step. 
\end{itemize}
\end{case}
\begin{case}[\ref{t-tick}] 
Since (H1) then we have that $\emptyset \typeOf e_1: T $, $\emptyset \typeOf e_2:\List\app\VarType$.
By IH on $e_1$, we have that 
\begin{itemize}
\item $e_1 = v_1$. By IH on $e_2$:
    \begin{itemize} 
    \item $e_2 = v_2$. Then we have two cases depending on $T$: 
        \begin{itemize} 
        \item $T =  \TermType$. By Canonical forms, we have that $e_1$ can be of the following forms:
            \begin{itemize} 
            \item $(opname \app v_1')$. Then we apply \ref{r-tick-opname} and take a step. 
            \item $\langVar{X}$. Then we apply \ref{r-tick-var} and take a step. 
            \item $(\langVar{z}v_1')$. Then we apply \ref{r-tick-abs} and take a step. 
            \item $v_1'[v_1''/\langVar{z}]$. Then we apply \ref{r-tick-sub} and take a step. 
            \end{itemize}
        \item $T = \List\app\TermType$. By Canonical form $e_1$ can have two forms:
            \begin{itemize} 
            \item $e_1 =  \nil$. Then we apply \ref{r-tick-nil} takes a step. 
            \item $e_1 = \cons{v_1}{v_2}$. Then we apply \ref{r-tick-cons} takes a step. 
            \end{itemize}
        \end{itemize}
    \item $e_2 = \errorLNC$. Then by \ref{ctx-err} we take a step to an error. 
    \item for all $V, \LangDef$, $V ; \LangDef ; e_2 \step V' ; \LangDef' ; e_2'$, for some $V',\LangDef',e_2'$. Then by \ref{ctx-succ} or \ref{ctx-lang-err}, we take a step. 
    \end{itemize}
    \item $e_1 = \errorLNC$. Then  by \ref{ctx-err} we take a step to an error. 
    \item for all $V, \LangDef$, $V ; \LangDef ; e_1 \step V' ; \LangDef' ; e_1'$, for some $V',\LangDef',e_1'$. Then by \ref{ctx-succ} or \ref{ctx-lang-err}, we take a step. 
\end{itemize}
\end{case}
\qed
\end{proof}

\begin{theorem}[Progress Theorem for Configurations]
For all , if $\emptyset \typeOf V ; \LangDef ; e$ then either 
\begin{itemize}
\item $e = \skipLNC$, or 
\item $e = \errorLNC$, or 
\item $V ; \LangDef ; e \step V' ; \LangDef' ; e'$, for some $e'$.
\end{itemize}
\end{theorem}
\begin{proof}
Let us assume the proviso: $\emptyset \typeOf V ; \LangDef ; e$. Then we have $\emptyset \typeOf  e : \LanguageType$.
By Progress Theorem for Expressions, we have that 
\begin{itemize}
\item $e = v$. By Canonical forms, $e = \skipLNC$.
\item $e = \errorLNC$, which satisfies the theorem. 
\item $V ; \LangDef ; e \step V' ; \LangDef' ; e'$, for some $e'$, which satisfies the theorem. 
\end{itemize}
\qed
\end{proof}

\end{document}
\newpage

$\langCom$ makes use of a domain-specific language to define languages in operational semantics. 
We do not provide a formal syntax for language definitions, as our domain-specific language is rather intuitive. 
Fig. \ref{fig:LNCrunexample} shows the language definition in $\langCom$ for $\runningExample$. 
\begin{figure*}[tbp]
{
\small
\lstinputlisting[numbers=left,language=lprolog]{lists.mod}
}
\caption{Example of language definition in $\langCom$. Binding is limited to unary lexical scoping \cite{Cheney:2005}, which is sufficient in the scope of our course. We express binding with syntax such as {\small\texttt{(x)E}}, that is, {\small\texttt{x}} is bound in {\small\texttt{E}}. This is similar to the directive {\small\texttt{(+ bind x in e +)}} in the Ott tool \cite{Ott}. \key{E[V/x]} represents the capture-avoiding substitution.} 
\label{fig:LNCrunexample}
\end{figure*}
%
This definition describes a grammar (lines 1-4), a type system (lines 6-12), and a dynamic semantics (lines 14-16), and does so with a textual representation of the syntax that is typically employed in operational semantics\footnote{This is not a novelty. The Ott language, for example, achieved the same effect over ten years ago \cite{Ott}.}. A few differences from pen\&paper formulations is that grammar variables are always capitalized and the horizontal line in an inference rule \[{\inference{{premises}}{conclusion}}\] is replaced with an inverse implication \texttt{<==} that can be read \emph{``provided that}". 


Language designers can define elimination forms with \emph{role directives} such as those at line 18 and 19. 
A directive \key{\# \app eliminator\app app\app arrow} means that the application operator is declared as elimination form for the function type.  

Language designers can also define the variance for types with \emph{variant directive} such as that at line 20. 
The directive \key{\#\app contravariant\app arrow\app 1} states that the first argument of the function type is contravariant. 
$\langCom$ also has the keywords \key{invariant} and \key{covariant}. When the variance of an argument is not declared it is taken to be covariant by default. 

Similar directives can be used to declare the modes for relations. 
$\langCom$
Modes are declared with \emph{mode directives} such as \key{mode:\app \textit{predicatename}\app inp\app inp\app \ldots\app out} (an output may or may not present). 
These directives do not appear in Fig. \ref{fig:LNCrunexample} because it makes use of relations that $\langCom$ considers builtin. 




\subsection{Expressions (Transformations)}

\paragraph{Main Components: Terms, Formulae, and Rules}
Terms have a uniform shape $(op\app e \ldots e)$ and they accommodate types, expressions, and other syntactic categories tha language designers may have. Examples of terms are $(\to\app T_1\app T_2)$ and $(\key{app}\app e_1\app e_2)$. 
The arguments of terms are general $\langCom$ expressions. We do so because we will often make use of other constructs of lang-n-change such as the if-then-else to build the arguments in a term. 
The syntax for terms also contains the capture-avoiding substitution $t[t/x]$ and binding $(x).t$ as primitive. 

Formulae are of the form $(predname\app e \ldots e)$, in which we use $\langCom$ expressions as arguments for the same reason we do for terms. 

Rules follow a standard shape with premises and conclusion. We allow to use general $\langCom$ expressions in both places. 

\paragraph{Lists and Maps}

$\langCom$ has lists and maps as primitive data types. Lists can be build with the usual operator $[\ldots]$ and can be accessed with $\head$ and $\tail$. Maps can be created with \key{createmap}. For example $\key{createMap}(T, [T_1, T_2]) \Rightarrow mymap: body$ creates the map with the key $T$ and the list the list [T1, T2] as its value and binds it to $mymap$ inside $body$. To operation $myMap(T)$ returns $[T_1, T_2]$. For the sake of simplicity, we allow only one key to be created by explicitly the programmer. Other operations may create maps with more keys, as point out below.

\paragraph{Basic Operations: Conditionals and Sequencing}

$\langCom$ has conditional operations in the form of an if-then-else operation. 
The guard of the if-then-else can be one of the boolean expressions $b$ in the grammar of Fig 2.
$\isContra(x)$ checks whether the type appears in contravariant position in the rule
Similarly, $\isCovar(x)$ and $\isInvar(x)$ check whether it appears in covariant or invariant position. 
$\overlap(x,x)$ takes two (bound) terms and check whether they have variables in common. 
$\inductive(x)$ checks whether the type is inductive, in the sense that the type has arguments that are types as well. 
For example, $\to$ and $\List$ are inductive but $\Int$ and $\Bool$ are not. 
$\isInput(x)$ check that a variable appears in input position in the current rule. $\eliminatorOf{x}{x}$ checks that an operator is declared as elimination form for a certain type. For example $\eliminatorOf{app}{\to}$ would hold and $\eliminatorOf{app}{\Bool}$ would be false. 
$\langCom$ also include two predicates on lists to check whether a list is empty and whether an element belongs to a list. 


The sequencing operator $e_1;e_2$ executes the two expression one after the other. The first expression $e_1$ transforms a language definition and returns another language definition. The second expression $e_2$ will then apply its transformations on this latter language definition. 

\paragraph{Selectors}
Selectors are one of the most important operations of $\langCom$. 
Selectors iterate a body for all instances of items that match a pattern. These items can be terms, formulae, or rules, depending on what the selector is applied for. 
Selectors are applied to a field $op$ which can be a string or a variable. 
Some strings have a specific purpose as we shall describe below. 

\key{rule}$[p]:e$ is used to select rules of the language whose conclusion matches the pattern $p$. In this case $p$ is a formula. As an example, \key{rule}$[Gamma |- e : T]$ selects all the typing rules. It can be used with more complex patterns, for example \key{rule}$[(\key{app }\app e_1\app e_2) --> e_3]$ selects all the reduction rules whose source is built with a top level operator \key{app}. Selectors perform ordinary pattern-matching and bindings are then available in the body $e$. The latter example would select the \textsc{(beta)} rule with the binding $e_1 = lambda. x:T.e$ and $e_2 = v$.

\key{premise}$[p]:e$ is used in the context of a rule and is used to select the premises of the rule that match the pattern $p$. To make an example, \key{premise}$[T_1 <: T_2]:e$ selects all the premises that use the subtyping relation. 

Selectors can be used to select items in grammars by mentioning the desired syntactic category. If we were to retrieve all the expressions, we would write expression$[op \app e_1 \app \ldots \app e_n]$. This has the effect to inspect the grammar for the  syntactic category Expression and try to match each grammar items with that pattern. 
This particular pattern uses the notation $\app e_1 \app \ldots \app e_n$ with dots, which is special notation in $\langCom$ to express that there can be a finite but not known number of arguments for $op$. For example, the pattern accommodates the if-then-else operator, which has 3 arguments, and also the application, which has 2 arguments. As a convention, body $e$ can refer collectively to $e_1 \app \ldots \app e_n$ as a list called $es$, that is, an $s$ is suffixed to the variable used to pattern-match a variable number of arguments. 
Grammar items do not have an index, such as $(\key{if} \app e\app e\app e)$ and $(\key{app}\app e\app e)$ but in most cases we makes use of these items from the grammar in the context of inference rules in which using the same variable actually means that they must be equal. Most of the time that is not the intended purpose. Therefore, $\langCom$ automatically extracts items from grammars giving them indexes. For example, we would extract $(\key{if} \app e_1\app e_2\app e_3)$ and $(\key{app}\app e_1\app e_2)$ from grammars.

$\langCom$ offers specific fields to select operators based on their roles. \key{operations}$[p]:e$ selects only operations, and \key{eliminators}$[p]:e$ selects only elimination forms.

We use selectors also for iterating over lists and maps. If $es$ is a list of terms we can use $es[e]:body$ with the meaning that each element of the list is bound to the variable $e$ in $body$. Similarly, if $mymap$ is a map, $mymap[k]:body$ iterates over all the keys of the map, binding each to $k$ in $body$.

In $\langCom$ we use the pattern $\star$ to simply select all the items without assigning a specific binding. 

Another way to use selectors is with the option \key{keep}. This has the effect to return the items that fail the pattern, in addition to apply the iteration to those that do match the pattern. As a convenience to programmer, the keyword \key{self} denotes the item being selected in the current iteration. For example \key{rule}$[\star]:$ \key{self} simply returns all the rules of the language.

\paragraph{Syntax Transformation}

$\langCom$ provides operations to the manipulating of grammars. Currently, we restricted $\langCom$ to two simple operations.  
The first syntax transformation can be used in the following way: Statements s ::= $(\key{while} \app b \app s)$. This has the effect to add this syntactic category to the language. If the category Statement existed then it is replaced. 
The second syntax transformation can be used as in Expression  e :: = $\ldots \mid (\key{cast} \app e\app T\app T)$. This has the effect to add the cast operator to the already existing grammar for expressions. If this syntactic category did not exist previously then it is created and the cast operator would be the only grammar item. 
In order to compute new grammar items we often makes use of other constructors of $\langCom$ such as if-then-else. 

\paragraph{Rule and Formula Operations}

$\uniquefy{x}{e}{}{}{}$ works in the context of a rule and assigns a new variable to each output that appears more than once with the same variable. To make an example, \key{uniquefyOuputs} turns \textsc{t-app} into 
\[
\inference
	{
	\Gamma \typeOf \app e_1 : T_3\to T_2  \qquad \Gamma \typeOf \app e_2 : T_4
	} 
	{ \Gamma \typeOf \app e_1\app e_2 : T_2}
\]
That is, $T_1$ has been split into new variables $T_3$ and $T_4$. Further transformations may need to know which variables have been split. To allow this, the operation also binds a map to $x$ in the body $e$. This map contains associations that describe the replacements taken place. For example, after this transformation on \textsc{t-app}, the map would contain $T$ as a key, and the list $[T_3,T_4]$ as its value.

\key{fold} takes a binary predicate name and a list of terms and creates a formula that relates the first element with the second, the second with the third, and so on. For example, $\key{fold}\app <:\app [T_1, T_2, T_3,T_4]$ results is the list of formulae $T_1 <: T_2$, $T_2 <: T_3$, and $T_3 <: T_4$.

\paragraph{Term Operations}

The keyword $\conclLNC$ can be used in the context of a rule and returns the conclusion of such rule. 
$\getOverlap{e_1}{e_2}{x}{e_3}$ computes the list of variables that $e_1$ and $e_2$ have in common, and binds this list to $x$ inside the body of $e_3$. 
$\newVar{x}{e}$ creates a new variable in the scope of $e$. $\langCom$ makes sure that this variables is being in use in the rule. 
The tick operation $\tickedLNC{e}$ is applied to a term and gives a prime to its variables. For example $(T)'$ returns the variable $T'$. $\langCom$ makes sure that $T'$ is not already used in the rule. If that is not the case another name is chosen. The tick operation also has a restricted version $\tickedMMLNC{e_1}{e_2}$, in which $e_2$ contains the list of variables in $e_2$ that are allowed to be ticked. To make an example, $\tickedMMLNC{(op\app e_1\app e_2\app e_3)}{[e2,e3]} = (op\app e_1\app e_2'\app e_3')$.

\section{Implementation and Examples}\label{examples}

\paragraph{Implementation}
We have implemented $\langCom$ in OCaml. 
Language definitions and language transformations are in different files. 
Language definitions are parsed into an OCaml data type representations. Transformations are compiled into OCaml code that manipulate that data type according to the instruction of the specified algorithm. 
Language definitions are then pretty-printed in the same format as language definitions in input. 

In this section, we employ $\langCom$ to describe algorithms that are of interest in programming languages. 

\paragraph{From small-step to big-step}

Our first example is a simple, though incomplete (see Section \ref{future}), algorithm to automatically transform a language definition from small step to big step semantics. 
The algorithm is below. 
{
\small
\lstinputlisting[numbers=left,language=lprolog]{big.tr}
}

Line 1 makes use of the selector for operations. We recall that the variables $e1$, $\ldots$, and  $en$ can be collectively be referred to as the list $es$. 
For each operation, we select the rules of the languages that obey the specific pattern in line 2.  This pattern selects only rules whole conclusion proves a $\step$-formula, and whose target has the top level operator $op$ of the current iteration. 
This has the effect to select the reduction rules that model the operational semantics of $op$. 
For each such rule, we create a new rule whose conclusion is in line 9. In the conclusion, the top level operator of the source is $op$ applied to the pairwise distinct variables $es$. The target can either be a new variable $X$ or the target that comes directly from the reduction rule selected. We will discuss this choice shortly. 
The premises of this rule are described by two parts. The first part is in lines 4-6, and the second is line 7. In lines 4-6 we strive to make each $e_i$ evaluate to $exp_i$, which is the expression that the small-step operational semantics rule prescribes to start with. To explain the reason for the tests in lines 4-5 we can take a look at the rule \textsc{(r-if-true)} after this transformation if we simply evaluated each $e_i$ to each $exp_i$. 
\[
\inference
	{
	exp_1 \step \true \\
	exp_2 \step e_2 \\
	exp_3 \step e_3 
	} 
	{\If \app exp_1 \app exp_2 \app exp_3 \step e_2}
\]
Here evaluating $e_3$ is useless because it does not contribute to the target. Therefore, our algorithm strives to be more precise and tests whether there are overlapping variables between that expression and the target in line 4. 
There is a caveaut to this principle, however. The principal argument in \textsc{(r-if-true)} is $\true$ and that too does not appear in the target! Nevertheless, since it is a complex expression we are supposed to get it in order to fire \textsc{(r-if-true)}. That is,  we need to evaluate $e_1$ to $\true$. Hence, the test at line 5. 
The premise is in synch with the computation of the target of line 8. For \textsc{(r-if-true)} we do not need a new variable because the target $e_2$ appears in $es$ and already has an evaluation step. For a rule such as \text{(beta)}, however, the target is complex and needs to be evaluated to yield the ultimate result. In that case we would make use of a new variable $X$, as follows. 
\[
\inference
	{
	exp_1 \step \lambda x:T.e \\
	exp_2 \step v \\
	e[v/x] \step X 
	} 
	{\If \app e_1 \app e_2 \app e_3 \step X}
\]

This transformation is composed with another at line 11. This latter transformation simply provides a reduction rule for all values. In big step semantics values evaluate to themselves.
Ultimately, if we apply lang-n-change to the language definition of Fig \ref{fig:LNCrunexample} we obtain the following (for brevity, only the operational semantics is displayed). 

{\small
\lstinputlisting[numbers=left,language=lprolog]{big.mod}
}

\paragraph{Automatic Subtyping}

Our second example is an algorithm for adding subtyping to languages. 
The algorithm is below. 

{
\small
\lstinputlisting[numbers=left,language=lprolog]{subtyping.tr}
}

The algorithm is a sequencing (;) of two selectors, the one that begins at Line 1 and the one that begins at line 28. Line 1 makes use of a selector for selecting rules. $\langCom$ makes a shorthand available for which the pattern $\typeOf$ simply means the patter $Gamma \typeOf \app e: T$, with free variables $Gamma$, $e$ and $T$. This has the effect to select all the typing rules. The body of the selector has only one instruction that is the uniquefy operation, which gives new variables to all outputs that are not distinct. After this operation, $myMap$ is in the scope of lines 3-24. We have a sequencing of 3 transformation expressions. The first expression is at lines 4-7. This transformation applies from the rule returned after the uniquefy operation. It keeps the same conclusion, and also keeps the same premises with line 4 (that is a $\langCom$ shorthand for \key{premise}$[*]: \app$\key{self}). Line 5 adds further premises: for each key in the map $myMap$ we relate the variables in the list in the value of that key with subtyping. 
To make an example, after the uniquefy operation and the transformation of lines 4-7, \textsc{(t-app)} would turn into the following rule. ($myMap$ has the only key $T_1$ mapped to the list $[T_3, T_4]$, which are the new variables that have replaced $T_1$).
\[
\inference
	{
	\Gamma \typeOf \app e_1 : T_3\to T_2 \\
	\Gamma \typeOf \app e_2 : T_4	\\
	T_3 <: T_4
	} 
	{ \Gamma \typeOf \app e_1\app e_2 : T_2}
\]
Lines 11-17 act on this rule. These lines specify a transformation in order to adjust the direction of the subtyping when needed. 
Line 11 simply keeps the typing rules as they are. Line 12, instead, selects the rules that speak about subtyping. For each of them, the two nested if-then-else operations perform the adjustment depending on the variance of types. For the rule above, we would have the swap $T_4 <: T_3$ from $T_3 <: T_4$ because $T_3$ is in contravariant position. 

Lines 21-24 perform a rule transformation that is responsible for handling the inputs that may have been left without an output after the uniquefy operation. Let us consider the typing rule for \key{if} outputs have been replaced. 
\[
\inference
	{
	\Gamma \typeOf \app e_1 :\Bool \\
	\Gamma \typeOf \app e_2 : T_1	\\
	\Gamma \typeOf \app e_3: T_2	
	} 
	{ \Gamma \typeOf \If e_1\app e_2\app e_3 : \HI{$T$}}
\]
In the conclusion, the input used to receive its output from the premises but now that $T$ has been replaced we have that that  variable is incorrectly free. The standard treatment (see TAPL \cite{}) is that $T$ receives the join of the two new variables $T_1$ and $T_2$, that is $T_1 \cap T_2 = T$. Line 22 does precisely so. For each key $T$ of then map if $T$ is also an input then we place in a premise that compute the join. The join predicate is such that the last argument is the output. At line 22, the join predicate does not take two arguments as it may look like. Mentioning the list $Map(T)$ expands to fill the arguments of join, which then are followed with the additional $T$ at the end. 

The second part of the algorithm is at line 28. This part is responsible for defining the subtyping relation. That is, is supposed to create the inference rules 
\begin{gather*}
\Int <: \Int \qquad \Bool <: \Bool\\
\inference
	{
	T_1' <: T_1 \qquad T_2 <: T_2'
	} 
	{ T_1 \to T_2  <: T_1' \to T_2' }
\end{gather*}
To do so, line 28 selects all the types from the syntactic grammar Types. Lines 29-34 creates a rule. The conclusion of this rule is a subtyping relation where we have the type on one side, applied to distinct variables, and essentially the same type on the other side but in which those variables are ticked. 
The premises of this rule are described in lines 29-33. For each variable $T$ in $Ts$, we put a subtyping premise that relates a variable and its ticked version in a pairwise fashion. The direction of subtyping depends on the variance of the type. 

When we apply the algorithm to the language definition of Fig \ref{fig:LNCrunexample}, $\langCom$ produces the following  language definition. 
(For brevity, only typing rules and subtyping rules are displayed). 

{
\small
\lstinputlisting[numbers=left,language=lprolog]{subtyping.mod}
}

Notice that $T_1$ and $T_2$ are not subject to variance because they are not in the context of any type constructor. 
Therefore, they are treated as invariant, yielding the double direction subtyping at lines 18 and 19 . 
These premises are redundant, as the existence of a join between $T_1$ and $T_2$ already implies $T_1 <: T_2$ and $T_2 <: T_1$ \cite{}. The algorithm could be modified to remove these premises. We show this patch in Section \ref{} when speaking of the limitations of the language. 

Notice also that the transformation in lines 28-34 takes care also of subtyping axiom cases for  $\Int$ and $\Bool$ because constant types simply are instances matched in line 28 in which we have $T_1 ... T_n$ with $n = 0$ arguments. 

\paragraph{Automatic Gradual Typing}

Gradual typing is an approach to programming languages in which static and dynamic typing are mixed in the same language. 
One of the most popular method to give a gradually typed language is to start with a statically typed version of a language. Language designers need to modify their language at hand in order to equip it with a treatment of dynamic typing. 
The first step is to add a specific type $\dyn$ to the language, with the meaning that expressions of type $\dyn$ are dynamically typed. The type system must be modified to be more liberal at the encounter of $\dyn$ because type checking is delegated to the run-time for them. 
Cimini and Siek offered an algorithm to automatically transform type system into their gradually typed version \cite{}. We have used $\langCom$ in order to implement that language transformation. 
As it turns out, most of the algorithm for introducing subtyping automatically is the same for gradualizing a type system. This is because subtyping too achieves an effect of forgoing comparing types by equality in order to apply a more liberal relation. We can see this point by comparing the following gradual typing rule for application. This typing rule is partially correct, as explained below. 
\[
\inference
	{
	\HI{$\Gamma \typeOf \app e_1 : T_3\to T_2$} \\
	\Gamma \typeOf \app e_2 : T_4 \qquad T3 \sim T4	
	} 
	{ \Gamma \typeOf \app e_1\app e_2 : T_2}
\]
As we can see this, the variable $T_1$ of \textsc{(t-app)} has been split into $T_3$ and $T_4$ and these two new variables have been related with a relation $\sim$ that is called type consistency. Type consistency has different axioms then subtyping in order to make the relation reflexive, symmetric but not transitive \cite{}, but has otherwise the same definition of subtyping (subtyping rules) modulo the use of $\sim$ where $<:$ occurs. We therefore omit showing the algorithm that generate $\sim$. 
An addition to that algorithm is, however, needed for the treatment of the premise highlighted above. 
Indeed, if we kept the typing rule with that premise we would impose our type system to type check succesfully only when the function argument really is of function type. In gradual typing, however, we must be prepared to type check successfully also when the fucntion is of dynamic type $\dyn$. Since $\dyn$ would not pattern-match with the output $T_3\to T_2$ required by the premise above, a legitimate simple program such as $\lambda x:\dyn. (x \app 4)$ would be rejected by the type system. The correct gradual typing rule for application, together with a needed auxiliary relation, is then the following. 
\begin{gather*}
\inference
	{
	\Gamma \typeOf \app e_1 : X \qquad X \matches T_1\to T_2 \\
	\Gamma \typeOf \app e_2 : T_3 \qquad T1 \sim T3 	
	} 
	{ \Gamma \typeOf \app e_1\app e_2 : T_2}\\[1ex]
\Int \matches \Int \qquad 
\dyn \matches \Int \\[1ex]
\Bool \matches \Bool \qquad
\dyn \matches \Bool \\[1ex]
T_1 \to T_2 \matches T_1 \to T_2 \qquad 
\dyn \matches \dyn \to \dyn  \\[1ex]
\List\app T\matches \List \app T \qquad 
\dyn \matches \List\app \dyn  
\end{gather*}
Here, we let the output of type checking $e_1$ be accommodated into a new variable $X$. Subsequently, we use a gradual matching operation to check that $X$ is either an ordinary function type, or a dynamic type \cite{}. Moreover, if the latter case occurs, the rest of the typing rule should continue pretending that the function at hand is, essentially, a function from dynamic type to dynamic type ($\dyn \to \dyn$). The treatment just described is not specific to the function type and is applied to most simple types such pairs, sums, lists, and so on. 
$\langCom$ can model this additional transformation by extending the algorithm for subtyping at line 24 with 

{
\small
\lstinputlisting[numbers=left,language=lprolog]{gradualStat.tr}
}

Line 1 selects all typing rules. For each of such rules, it creates a new rule in which all premises that pattern-match the pattern $Gamma |- e : (c \app T1\app ...\app T3)$ are subject to the following transformation at lines 3-5 (notice that the pattern searches for an output that is a complex type). The premises that fail this pattern are simply kept as they are. 
For the selected premises, we create a new variable $X$ and we apply the treatment described previously. 
Afterwards, we generate the definition of gradual matching at lines 9-12. 
We select all types at line 9. Line 10 models the fact that a type must matches itself, as in $T_1 \to T_2 \matches T_1 \to T_2$. Line 12 models the part in which the dynamic type is allowed to match a type. In that case, all the arguments of the type must be the dynamic type and we can put the correct number of $\dyn$ thanks to $(Ts[*]: \dyn)$, that means: For each argument put a $\dyn$.
Ultimately, line 16 adds the dynamic type to its grammar for types.

A complicated aspect of gradual typing is the execution of gradual programs. 
For gradual programs, the compiler would insert run-time checks whenever the type system has been optimistic and let the program type check successfully. 
Casts are a standard mechanisms for implementing run-time type checks, and they have been adopted widely in gradual typing \cite{}.  
With this choice, a language designer must equip the language at hand with casts and with a suitable operational semantics for casts. 
The following is the algorithm. 
\begin{figure*}[tbp]
{
\small
\lstinputlisting[numbers=left,language=lprolog]{gradualDyn.tr}
}
\caption{Algorithm for the Core Part of the Gradual Dynamic Demantics.}
\label{fig:gradDynAlg}
\end{figure*}

Figure \ref{fig:gradDynAlg} contains the algorithm for computing the core part of the gradual dynamic semantics. Line 1 augments the language with a cast operator. $\castLNC \app e \app T_1 \app T_2$ means that $e$ is cast from $T_1$ to $T_2$. 
The standard cast semantics strictly disciplines what types can be injected into the dynamic type \cite{}. In particular, type constructors applied to only $\dyn$ can be injected into $\dyn$. These types are called ground types. To make an example, we can $\dyn \to \dyn$ into $\dyn$ but we cannot cast $\Int \to \Int$ directly into $\dyn$. (This cast would be handled with an intermediate burocratic step, see \cite{}).
Line 4 adds ground types to the grammar accordingly. 

Casts such as $\castLNC \app \lambda x:T.e : \Int \to \dyn \Rightarrow \dyn \to \Int$ are considered values. This is because we do not have a means to know if $\lambda x:T.e$, in which $e$ is dynamically typed, will actually return an integer. We would have to apply this function to some argument and check that. 
Therefore, the grammar for values must reflect this. It is also important to mention that the scenario described applies only to types that have arguments, which we call inductive, such as $\rightarrow$, $\List$, but not base types such as $\Int$ and $\Bool$. For the latter types, indeed, casts to and from can be easily checks. 
Accordingly, lines 6-8 is responsible to add values that are wrapped in casts and only for inductive types. 
Line 6 augments the already existing grammar of values with these new values.

The core of the algorithm is at lines 8-16. This code is responsible for producing new reduction rules that gives operators the ability to handle casts. 
To make an example, the application operator must be equipped with the following reduction rule \cite{}. 

$$
\begin{array}{c}
((\castLNC\app v_1 \app (T_1' \to T_2') \app (T_1 \to T_2)) \app v_2)\\
\step \\
(\castLNC \app (v_1 \app (\castLNC \app v_2 \app T_1 \to T_1')) \app T_2' \to T_2)
\end{array}
$$
Since we have a function of type $T_1' \to T_2'$ but we need to simulate that we have a function of type $T_1 \to T_2$, we cast the argument from $T_1$ to $T_1'$ before we pass it. Furthermore, we cast the result from $T_2'$ to $T_2$. 

To automatically build this rule, and to do so for other common operations, we follow the schema of the Gradualizer \cite{}. 
We first select the eliminationForms (line 8)
For the source of the reduction rule, we place a cast at the principal argument. 
For The second thing to notice is that in one step casts are removed from the principal argument. This is because, ultimately, the function that is underneath the casts is the only value that we have to perform a step. We liberate casts after casts until we can use the ordinary beta reduction. In removing casts, however, we have redistributed them to the other argument, which we call the sibling of the principal argument, and in wrapping the whole expression. 
Redistributing casts is necessary to make sure that the reduction step is type preserving. 
To do so, we need to look at the typing rule for the application. $e_2$ makes use of the types of the principal argument. This means that when the cast is removed, it now needs to align with the set of types that are used by the function that is underneath the cast, hence the cast on the sibling argument. 
Here, since 

The transformation that applies 


%
%
%

\section{Limitations and Future Work}\label{future}

One of the limitations is that selectors could be more potent. 
It would be interesting to add a search such as to search all the rules. 
Right now this is possible by scanning rules and pattern-matching $<:$. Turning a subtyping into a gradual typing could be just by saying "turn all <: into $\sim$". 
Right now the language does not have to bing the components of the conclusion of rules arbitrarily. 
Right now, we refer to those parts but if nested patterns are needed we are at loss.  It would be interesting if the programmer could so that it could use the compoennts. 

At the moment, there is no simple way to simply lift a relation to another shape. For example, if we were to add references to a language like the lambda-calculus and similar automatically, we would have to shift from a typing relation to , and from a reduction relation to . It would be interesting to have operators such as lift. 

operation to delete stuff.. 
give a name to rules and select by name such as r-* all the reduction rules. 

Some future work that we plan are inspired with the related work and then discussed in the next section.

%

\section{Related Work}\label{related}

\section{Conclusion}\label{conclusion}

\appendix
\section{Proof of Type Soundness}

\begin{theorem}[Subject Reduction]
For all $\Gamma$, $V$, $V'$, $\LangDef$, $\LangDef'$, $e$, $e'$, if $\Gamma \typeOf V ; \LangDef ; e$ and $V ; \LangDef ; e \step V' ; \LangDef' ; e'$ then $\Gamma \typeOf V' ; \LangDef' ; e'$ and $V\subseteq V'$.
\end{theorem}

\begin{theorem}[Progress Theorem]
For all , if $\Gamma \typeOf V ; \LangDef ; e$ then either 
\begin{itemize}
\item $e = \skipLNC$, or 
\item $e = \errorLNC$, or 
\item $V ; \LangDef ; e \step V' ; \LangDef' ; e'$, for some $e'$.
\end{itemize}
\end{theorem}

\begin{theorem}[Type Soundness]
For all $\Gamma$, $V$, $\LangDef$, $e$, if $\Gamma \typeOf V ; \LangDef ; e$ then $V ; \LangDef ; e \step^{*} V' ; \LangDef' ; e'$ such that  
\begin{itemize}
\item $e' = \skipLNC$, or 
\item $e' = \errorLNC$, or 
\item $V' ; \LangDef' ; e' \step V'' ; \LangDef'' ; e''$, for some $e''$.
\end{itemize}
\end{theorem}

\end{document}

\begin{figure}[tbp]
\small
\textsf{Syntax}  \hfill
\begin{syntax}
  &&&\\
  \text{\sf Basic Types} & B & ::= &  \Int   \\
  \text{\sf Types} & T & ::= &  B \mid \app T\to  T  \\
  \text{\sf Expressions} & e & ::= &  x \mid \lambda x:T.e\mid e\app e  \\
   \text{\sf Values} & v & ::= & \lambda x.e  \\
  \text{\sf Contexts} & E & ::=  &  E\app e\mid v\app E   \\
   \end{syntax}\\
   
\textsf{Type System}  \hfill  \fbox{$\Gamma \vdash e : T$}

\begin{gather*}
\ninference{t-var}{}{
\Gamma\typeOf \app x: \Gamma(x)}
\qquad 
\ninference{t-lambda}
	{\Gamma,x:T_1 \typeOf \app e : T_2} 
	{ \Gamma \typeOf \app \lambda x:T_1. e:  T_1\to T_2}
\\[2ex]
\ninference{t-app}
	{
	\Gamma \typeOf \app e_1 : T_1\to T_2 \\
	\Gamma \typeOf \app e_2 : T_1 	
	} 
	{ \Gamma \typeOf \app e_1\app e_2 : T_2}
\end{gather*}
\textsf{Dynamic Semantics}  \hfill \fbox{$e\step e$}

\begin{align*}
	(\lambda x:T. e)\app v & \step  e[v/x]    	\label{beta}\tagsc{beta}
\end{align*}
$$
\inference
	{e\step e'}
	{E[e] \step E[e']}\,\,\textsc{(ctx)}
\qquad\qquad 
E[\mathit{er}] \step \mathit{er}\,\,\textsc{(err-ctx)}
$$

\caption{Simply typed lambda-calculus.}
\label{fig:LambdaFull}
\end{figure}

\begin{figure}[tbp]
\small
\begin{syntax}
  &&&\\
  \text{\sf Basic Types} & B & ::= &  \Int  \\
  \text{\sf Types} & T & ::= &  B \mid \app T\to  T  \\
  \text{\sf Expressions} & e & ::= &  x \mid \lambda x:T.e\mid e\app e  \\
   \text{\sf Values} & v & ::= & \lambda x.e  \\
  \text{\sf Contexts} & E & ::=  &  E\app e\mid v\app E   \\
   \end{syntax}\\
   
\textsf{Type System}  \hfill  \fbox{$\Gamma \vdash e : T$}

\begin{gather*}
\ninference{t-var}{}{
\Gamma\typeOf \app x: \Gamma(x)}
\qquad 
\ninference{t-lambda}
	{\Gamma,x:T_1 \typeOf \app e : T_2} 
	{ \Gamma \typeOf \app \lambda x:T_1. e:  T_1\to T_2}
\\[2ex]
\ninference{t-app}
	{
	\Gamma \typeOf \app e_1 : T_1\to T_2 \\
	\Gamma \typeOf \app e_2 : \HI{$T_3$} \\ \HI{$T_3 <: T_1$}
	} 
	{ \Gamma \typeOf \app e_1\app e_2 : T_2}
\end{gather*}
\textsf{Subtyping}  \hfill  \fbox{$T<:T$}
\begin{gather*}
\HI{$\Int <: \key{Float}$}
\qquad 
\HI{$\ninference{s-arrow}
	{
	 \typeOf T_3<: T_1 \\
	 \typeOf T_2<: T _4	
}
	{ \typeOf T_1\to T_2 <: T_3\to T_4}
$}
\end{gather*}
\textsf{Dynamic Semantics}  \hfill \fbox{$e\step e$}
\begin{align*}
	(\lambda x:T. e)\app v & \step  e[v/x]    	\label{beta}\tagsc{beta}
\end{align*}
$$
\inference
	{e\step e'}
	{E[e] \step E[e']}\,\,\textsc{(ctx)}
\qquad\qquad 
E[\mathit{er}] \step \mathit{er}\,\,\textsc{(err-ctx)}
$$

\caption{Simply typed lambda-calculus with subtyping.}
\label{fig:LambdaFull}
\end{figure}

\begin{figure}[tbp]
\small

$\mathcal{L}^{<:} = (\textsf{Syntax}, \textit{addSub}(\textsf{Type System}), \textsf{Dynamic Semantics},  \textit{subtyping}(\textsf{Syntax}))$ \\[2ex] %
$\textit{addSub}(TS) = $\\
\hspace*{0.5cm}$\forLang{r}{TS} 
\hspace*{1.4cm}$\inference{r.\key{premises.uniqueOutput} \Rightarrow \textit{Vars}}{r.\key{conclusion}} \app \key{as} \app r' \app \key{in}\\[2ex]
\hspace*{3cm}\inference{
\hspace*{-3.5cm}\key{if}\app {\key{exists\textendash contra}(r',\textit{Vars})} \\\\ 
\hspace*{1cm}\key{then} \app {\forLang{T}{\textit{Vars}}'T<: \key{contra}(r',\textit{Vars})}\\\\ 
\hspace*{0.3cm}\key{else} \app{\forOrderedLang{T_1}{T_2}{\textit{Vars}}'T_1<:T_2}}{r'.\key{conclusion}} 
$\textit{subtyping}(S) = ['\Int <: \Float,$ \\
$\hspace*{2.4cm}\forLang{c}{S.\key{types}}$
$$
\hspace*{2.4cm} \inference{
{\begin{array}{ccc}
\hspace*{-5.3cm}\forLang{\textit{arg}}{\textit{Args}} \\
\hspace*{-1.6cm}\key{if} \app c.\key{contra}(\textit{arg}) \app \key{then} \app\app '\key{primed}(\textit{arg}) <:  \textit{arg} \app \\
\hspace*{-1.1cm}\key{elseif}  \app c.\key{covar}(\textit{arg}) \app \key{then} \app\app '\textit{arg} <: \key{primed}(\textit{arg}) \app \\
\hspace*{-5.4cm}\key{else}  \app\app '\textit{arg} = \textit{arg}
\end{array}
}}{(c.\key{opName}\app (c.\key{newArgs} \Rightarrow \textit{Args})) <: (c.\key{opName}\app c.\key{newArgsPrimed})}$$ 
\hspace*{2.3cm}]
\caption{Transformation of $\mathcal{L}$ into $\mathcal{L}^{<:}$ with $\langCom$}
\label{fig:LambdaFull}
\end{figure}

\begin{figure}[tbp]
\small

$\mathcal{L}^\textsf{big} = (\textsf{Syntax}, \textsf{Type System}, \textit{turnToBig}(\textsf{Syntax},\textsf{Dynamic Semantics}))$ \\[2ex] %
$\textit{turnToBig}(S,DS) = $\\
\hspace*{0.5cm}$\forLang{v}{S.\key{values}} \inference{}{v \step v}$
\hspace*{4cm}$++$\\[2ex]
\hspace*{0.5cm}$\forLang{r}{DS} $\\
\hspace*{1cm}$\key{let}\app \textit{topOp} = r.\key{conclusion}.1.\key{topOp}\app \key{in }$
\hspace*{1cm}$\key{let}\app \textit{args} = r.\key{conclusion}.1.\key{args}\app \key{in }$
\hspace*{1cm}$\key{if}\app \textit{topOp} = 'E\app \key{then} \app \key{skip}$
\hspace*{1cm}$\key{else} $
\hspace*{2cm}$\inference{
\hspace*{-3cm}\key{let}\app \textit{hashmap} = \app args.\key{assign\textendash newVars} \app \key{in} \\\\
\forLang{(arg, var)}{hashmap} 'arg \step var \\\\++\\\\ [r.\key{conclusion}.2 \step r.\key{newVar]} 
}{r.\key{conclusion}.1 \step r.\key{newVar}}$ 
\caption{Transformation of $\mathcal{L}$ into $\mathcal{L}^{\textsf{big}}$ with $\langCom$}
\label{fig:LambdaFull}
\end{figure}

\begin{figure}[tbp]
\small
\begin{syntax}
  &&&\\
  \text{\sf Basic Types} & B & ::= &  \Int \mid \app \Float  \\
  \text{\sf Types} & T & ::= &  B \mid \app T\to  T  \\
  \text{\sf Expressions} & e & ::= &  x \mid \lambda x:T.e\mid e\app e  \\
   \text{\sf Values} & v & ::= & \lambda x.e  \\
  \text{\sf Contexts} & E & ::=  &  E\app e\mid v\app E   \\
   \end{syntax}\\
   
\textsf{Type System}  \hfill  \fbox{$\Gamma \vdash e : T$}

\begin{gather*}
\ninference{t-var}{}{
\Gamma\typeOf \app x: \Gamma(x)}
\qquad 
\ninference{t-lambda}
	{\Gamma,x:T_1 \typeOf \app e : T_2} 
	{ \Gamma \typeOf \app \lambda x:T_1. e:  T_1\to T_2}
\\[2ex]
\ninference{t-app}
	{
	\Gamma \typeOf \app e_1 : T_1\to T_2 \\
	\Gamma \typeOf \app e_2 : T_1 	
	} 
	{ \Gamma \typeOf \app e_1\app e_2 : T_2}
\end{gather*}
\textsf{Dynamic Semantics}  \hfill \fbox{$e\step e$}

\begin{gather*}
\HI{$
	(\lambda x:T. e) \step  (\lambda x:T. e)    	
$}
	\\[2ex]
\HI{$
\ninference{beta-big}
{e_1 \step (\lambda x:T. e) \\ 
e_2 \step v\\
e[v/x] \step e'
}
{e_1\app e_2 \step  e'}
$}
\end{gather*}
\caption{Simply typed lambda-calculus in big-step semantics.}
\label{fig:LambdaFull}
\end{figure}

\bibliographystyle{ACM-Reference-Format}
\bibliography{all.bib}

\end{document}

\end{document}

\section{First Section}
\subsection{A Subsection Sample}
Please note that the first paragraph of a section or subsection is
not indented. The first paragraph that follows a table, figure,
equation etc. does not need an indent, either.

Subsequent paragraphs, however, are indented.

\subsubsection{Sample Heading (Third Level)} Only two levels of
headings should be numbered. Lower level headings remain unnumbered;
they are formatted as run-in headings.

\paragraph{Sample Heading (Fourth Level)}
The contribution should contain no more than four levels of
headings. Table~\ref{tab1} gives a summary of all heading levels.

\begin{table}
\caption{Table captions should be placed above the
tables.}\label{tab1}
\begin{tabular}{|l|l|l|}
\hline
Heading level &  Example & Font size and style\\
\hline
Title (centered) &  {\Large\bfseries Lecture Notes} & 14 point, bold\\
1st-level heading &  {\large\bfseries 1 Introduction} & 12 point, bold\\
2nd-level heading & {\bfseries 2.1 Printing Area} & 10 point, bold\\
3rd-level heading & {\bfseries Run-in Heading in Bold.} Text follows & 10 point, bold\\
4th-level heading & {\itshape Lowest Level Heading.} Text follows & 10 point, italic\\
\hline
\end{tabular}
\end{table}

\noindent Displayed equations are centered and set on a separate
line.
\begin{equation}
x + y = z
\end{equation}
Please try to avoid rasterized images for line-art diagrams and
schemas. Whenever possible, use vector graphics instead (see
Fig.~\ref{fig1}).

\begin{figure}
\includegraphics[width=\textwidth]{fig1.eps}
\caption{A figure caption is always placed below the illustration.
Please note that short captions are centered, while long ones are
justified by the macro package automatically.} \label{fig1}
\end{figure}

\begin{theorem}
This is a sample theorem. The run-in heading is set in bold, while
the following text appears in italics. Definitions, lemmas,
propositions, and corollaries are styled the same way.
\end{theorem}
%
%
\begin{proof}
Proofs, examples, and remarks have the initial word in italics,
while the following text appears in normal font.
\end{proof}
For citations of references, we prefer the use of square brackets
and consecutive numbers. Citations using labels or the author/year
convention are also acceptable. The following bibliography provides
a sample reference list with entries for journal
articles~\cite{ref_article1}, an LNCS chapter~\cite{ref_lncs1}, a
book~\cite{ref_book1}, proceedings without editors~\cite{ref_proc1},
and a homepage~\cite{ref_url1}. Multiple citations are grouped
\cite{ref_article1,ref_lncs1,ref_book1},
\cite{ref_article1,ref_book1,ref_proc1,ref_url1}.
%
%
%
%

\end{document}